\documentclass[12pt]{amsart} 

\usepackage[utf8]{inputenc}
\usepackage[T1]{fontenc}	
\usepackage[french,english]{babel}

\usepackage{lmodern}			
\usepackage{newtxtext}

\usepackage[top=3cm, bottom=3cm, left=3.7cm, right=3.7cm]{geometry}
\usepackage[shortlabels]{enumitem}

\usepackage{graphicx}	
\usepackage{tikz}
\usetikzlibrary{patterns}

\usepackage[a4paper,plainpages=false,colorlinks,linkcolor=bleuFonce,citecolor=rougeFonce,urlcolor=vertFonce,breaklinks]{hyperref}

\usepackage{placeins}
\usepackage{float} 

\usepackage{color}
\definecolor{vertFonce}	{rgb}{0,0.5,0}
\definecolor{numLignes}	{rgb}{0.17,0.57,0.7}	
\definecolor{gris}		{rgb}{0.5,0.5,0.5}
\definecolor{grisFonce}	{rgb}{0.2,0.2,0.2}
\definecolor{orange}	{rgb}{1,0.65,0.31}		
\definecolor{orangeFonce}{rgb}{1,0.4,0}
\definecolor{bleuFonce}	{rgb}{0,0,0.4}
\definecolor{rougeFonce}{rgb}{0.3,0,0}
\definecolor{rougeWord}	{rgb}{0.5,0,0}
\definecolor{vertClair}	{rgb}{0.8,1,0.8}
\definecolor{rougeClair}{rgb}{1,0.5,0.5}
\definecolor{violet}	{rgb}{0.5,0,0.5}

\usepackage{pict2e}
\usepackage{multido}

\usepackage{amsfonts,amssymb,amsthm,amsmath}
\usepackage{amsaddr}		
\usepackage{dsfont}			
\usepackage{cancel}
\usepackage{slashed} 
\usepackage{yfonts}
 
\newtheorem{theorem}{Theorem}[section]
\newtheorem{cor}{Corollary}[section]
\newtheorem{prop}{Proposition}[section]

\newtheorem{remark}{Remark}[section]

\newenvironment{system*}{%
	\equation\nonumber\left\{\ \begin{aligned}
}{%
	\end{aligned} \right. \endequation%
}

\newcommand		{\N}		{\mathbb N}			

\newcommand		{\RR}		{\mathbb R}			
\newcommand		{\R}		{\RR}
\newcommand		{\Rd}		{\R^d}
\newcommand		{\Rdd}		{\R^{2d}}

\newcommand		{\CC}		{\mathbb C}			

\newcommand		{\cE}		{\mathcal E}

\newcommand		{\cH}		{\mathcal H}		

\newcommand		{\cL}		{\mathcal L}
\newcommand		{\cN}		{\mathcal N}		
\newcommand		{\opP}		{\mathcal P}		
\renewcommand	{\L}		{\mathcal L}		

\newcommand		{\cB}		{\mathcal B}		

\newcommand		\sfA		{\mathsf A}

\newcommand		\sfH		{\mathsf H}			

\newcommand		\sfR		{\mathsf R}			
\newcommand		\sfS		{\mathsf S}			

\newcommand		\sfX		{\mathsf X}			

\newcommand		{\lt}			{\left}				%
\newcommand		{\rt}			{\right}			%
\renewcommand	{\(}			{\lt(}
\renewcommand	{\)}			{\rt)}

\newcommand		{\bangle}[1]	{\lt\langle #1\rt\rangle}

\newcommand		{\com}[1]		{\lt[{#1}\rt]}		

\newcommand		{\n}[1]			{\lt\lvert #1 \rt\rvert}
\newcommand		{\nrm}[1]		{\lt\lVert #1 \rt\rVert}
\newcommand		{\snrm}[1]		{\lVert #1 \rVert}

\newcommand		{\Nrm}[2]		{\nrm{#1}_{#2}}

\renewcommand		{\d}		{\mathop{}\!\mathrm{d}}	
			
\newcommand			{\dpt}		{\partial_t}

\newcommand			{\dt}		{\frac{\d}{\d t}}	

\newcommand			{\grad}		{\nabla}
\newcommand			{\lapl}		{\Delta}
\newcommand			{\Dx}		{\nabla_x}
\newcommand			{\Dchi}		{\nabla_\chi}
\newcommand			{\Dv}		{\nabla_\xi}
\newcommand			{\conj}[1]	{\overline{#1}}		

\DeclareMathOperator{\cF}		{\mathcal{F}}		
\DeclareMathOperator{\re}		{Re}				
\DeclareMathOperator{\im}		{Im}				
\DeclareMathOperator{\tr}		{Tr}				
\DeclareMathOperator{\diag}		{diag}

\newcommand		{\Tr}[1]		{\tr\!\( #1 \)}		

\newcommand		{\Diag}[1]		{\diag\!\( #1 \)}

\newcommand		{\intd}			{\int_{\Rd}}
\newcommand		{\intdd}		{\int_{\Rdd}}
\newcommand		{\intgam}		{\int_{\Rdd}}
\newcommand		{\intggam}		{\iint_{\Rdd\times\Rdd}}
\newcommand		{\iintd}		{\iint_{\Rdd}}

\newcommand		{\jj}			{\mathrm{j}}	
\newcommand		{\init}			{\mathrm{in}}

\newcommand		{\Eps}			{\mathcal{E}}
\newcommand		{\cC}			{\mathcal{C}}

\usepackage{braket}
\newcommand		{\Inprod}[2]	{\Braket{#1 | #2}}

\newcommand		{\op}		{\boldsymbol{\rho}}	

\newcommand		{\opUp}		{\boldsymbol{\Upsilon}}
\newcommand		{\opmu}		{\boldsymbol{\mu}}	

\newcommand		{\opnu}		{\boldsymbol{\nu}}	
\newcommand		{\opgam}	{\boldsymbol{\gamma}}
\newcommand		{\opc}		{\boldsymbol{c}}
\newcommand		{\opC}		{\boldsymbol{C}}

\DeclareMathOperator{\W}	{W}
\newcommand		{\Wh}		{\W_{2,\hbar}}		

\newcommand		{\opp}		{\boldsymbol{p}}

\newcommand		{\h}		{\mathfrak{h}}		
\newcommand		{\hh}		{\mathfrak{H}}		

\newcommand{\al}	{\alpha}

\newcommand{\ome}	{\op}

\newcommand{\id}{\mathbf{1}}
\newcommand{\x}{\mathbf{x}}
\newcommand{\y}{\mathbf{y}}

\newcommand{\updownarrows}{\mathbin\uparrow\downarrow}
\newcommand{\downuparrows}{\mathbin\downarrow\uparrow}
\renewcommand{\upuparrows}{\mathbin\uparrow\uparrow}
\renewcommand{\downdownarrows}{\mathbin\downarrow\downarrow}

\numberwithin{equation}{section}

\title[\textsc{Semiclassical Limit of the Bogoliubov--de Gennes Equation}]{\Large Semiclassical Limit of the Bogoliubov--de Gennes Equation}

\author[\textsc{J.J.~Chong}]{\vspace{-5pt}\large\textsc{Jacky J. Chong}} 
\address[J.J.~Chong]{\vspace{-10pt}School of Mathematical Sciences,\\ Peking University, Beijing, China\vspace{-15pt}}
\email{jwchong@math.pku.edu.cn}

\author[\textsc{L.~Lafleche}]{\large\textsc{Laurent Lafleche}} 
\address[L.~Lafleche]{\vspace{-10pt}Unit\'e de Math\'ematiques pures et appliqu\'ees,\\ \'Ecole Normale Supérieure de Lyon, Lyon, France\vspace{-15pt}}
\email{laurent.lafleche@ens-lyon.fr}

\author[\textsc{C.~Saffirio}]{\large\textsc{Chiara Saffirio}} 
\address[C.~Saffirio]{\vspace{-10pt}Department of Mathematics and Computer Science,\\ University of Basel, Basel, Switzerland}
\email{chiara.saffirio@unibas.ch}


\begin{document}

\begin{abstract}
	In this paper, we rewrite the time-dependent Bogoliubov--de Gennes equation in an appropriate semiclassical form and establish its semiclassical limit to a two-particle kinetic transport equation with an effective mean-field background potential satisfying the one-particle Vlasov equation. Moreover, for some semiclassical regimes, we obtain a higher-order correction to the two-particle kinetic transport equation, capturing a nontrivial two-body interaction effect. The convergence is proven for $C^2$ interaction potentials in terms of a semiclassical optimal transport pseudo-metric. Furthermore, combining our current results with the results of Marcantoni et al.~[arXiv:2310.15280], we establish a joint semiclassical and mean-field approximation of the dynamics of a system of spin-$\frac{1}{2}$ Fermions by the Vlasov equation in some negative order Sobolev topology. 
\end{abstract}

\begingroup
\def\uppercasenonmath#1{} 
\let\MakeUppercase\relax 
\maketitle
\thispagestyle{empty} 
\endgroup

\renewcommand{\contentsname}{\centerline{Table of Contents}}
\setcounter{tocdepth}{2}	
\tableofcontents


\section{Introduction}

\subsection{The time-dependent Bogoliubov--de Gennes equation} 

	We consider the time-dependent Bogoliubov--de Gennes (BdG) equation, sometimes referred to as the generalized Hartree--Fock equation or the Hartree--Fock--Bogoliubov equation. It describes the time evolution of generalized one-particle reduced density operators, which are self-adjoint operators $\Gamma$ acting on $\h\oplus \h$, satisfying the operator bound $0\le \Gamma \le \id_{\h\oplus\h}$, and having the form
	\begin{equation}\label{eq:expression_Gamma}
	\Gamma = 
	\begin{pmatrix}
		\gamma & \alpha\\
		-\conj{\alpha} & \id-\conj{\gamma}
	\end{pmatrix},
	\end{equation}
	where $\h=L^2(\Rd,\CC)$ denotes the one-particle state space, $\conj{A}$ represents the operator whose integral kernel is the conjugate of the kernel of the operator $A$. Here, the bounded linear operators $\gamma$ and $\alpha$ acting on $\h$ are called, respectively, the one-particle density operator and the pairing operator. It follows from properties of $\Gamma$ that $\gamma$ and $\alpha$ satisfy
	\begin{equation}\label{eq:eq_alpha_0}
		0\le \gamma \le \id, \quad \alpha^\ast = -\conj{\alpha},\quad \gamma\,\alpha = \alpha\,\conj{\gamma}, \quad \text{ and } \quad \n{\alpha^\ast}^2 \leq \gamma \(\id -\gamma\),
	\end{equation}
	where $\n{A}= \sqrt{A^\ast A}$.
	
	The BdG equation models many-body dynamics with an interparticle interaction potential $U$, which we assume to satisfy, for some constant $C\ge 0$, the conditions
	\begin{equation}\label{conditions:interaction_potential_U}
		U(x)=U(-x), \qquad U \in L^2_{\mathrm{loc}}(\Rd), \quad \text{ and } \quad U^2 \le C \(\id -\lapl\).
	\end{equation}
	Let $\opp = -i\hbar\grad_x$  denote the momentum operator with $\hbar=h/(2\pi)$, where $h$ is the Planck constant, and define the Hartree--Fock Hamiltonian associated with $\gamma$ by the operator
	\begin{equation*}
		\sfH_\gamma = \frac{\n{\opp}^2}{2} + V_\gamma -\sfX_{\gamma}.
	\end{equation*}
	In the above expression, $V_\gamma$ is the multiplication operator by the mean-field potential defined by
	\begin{equation*}
		V_\gamma(x) = U*\diag(\gamma)(x) = \intd U(x-y)\,\gamma(y, y)\d y
	\end{equation*}
	and $\sfX_{\gamma}$ is the exchange operator defined through its integral kernel 
	\begin{equation*}
		\sfX_{\gamma}(x, y) = U(x-y)\,\gamma(x , y ).
	\end{equation*}
	Moreover, we define the generalized Hartree--Fock Hamiltionian acting on $\h\oplus \h$ by the matrix operator
	\begin{equation*}
		\sfH_\Gamma =
		\(\begin{array}{cc}
			\sfH_\gamma & \sfX_\alpha\\
			\sfX_\alpha^\ast & -\conj{\sfH}_\gamma
		\end{array}\)\, .
	\end{equation*}
	Then the BdG equation reads
	\begin{equation}\label{eq:BdG_Gamma}
		i\hbar\,\dpt \Gamma = \com{\sfH_{\Gamma},\Gamma},
	\end{equation}
	where $\com{A, B}:=AB-BA$ is the operator commutator. 
	Equivalently, Equation \eqref{eq:BdG_Gamma} can be written as the following coupled system of equations
	\begin{subequations}\label{eq:BdG}
	\begin{align}\label{eq:BdG_gamma}
			i\hbar\,\dpt \gamma &= \com{\sfH_{\gamma},\gamma} + \sfX_\alpha\,\alpha^* - \alpha\,\sfX_\alpha^* \, ,
			\\\label{eq:BdG_alpha}
			i\hbar\,\dpt \alpha &= \sfH_{\gamma}\,\alpha + \alpha \,\conj{\sfH}_{\gamma} + \sfX_\alpha \(\id-\conj{\gamma}\)-\gamma \,\sfX_\alpha\, .
	\end{align}
	\end{subequations}
	Notice that if $\alpha = 0$, then $\gamma$ solves the Hartree--Fock equation
	\begin{equation}\label{eq:unscaled_Hartree-Fock}
		i\hbar\,\dpt \gamma = \com{\sfH_{\gamma},\gamma} .
	\end{equation}
	This justifies the claim that the BdG equation is a generalization of the Hartree--Fock equation with a non-zero pairing operator. Observe also the fact that the self-adjointness of $\Gamma$ and $0 \leq \Gamma \leq \id_{\h\oplus\h}$ are preserved under the BdG dynamics, in particular the Properties~\eqref{eq:eq_alpha_0} remain true along the dynamics. We refer to~\cite{benedikter_diracfrenkel_2018} for discussions regarding the well-posedness and additional properties of the equation. In fact, Conditions \eqref{conditions:interaction_potential_U} are taken from \cite{benedikter_diracfrenkel_2018}, which guarantees the global well-posedness of solutions. 
		
\subsection{Semiclassical regimes and classical phase space dynamics}

	The purpose of this paper is the study of the BdG equation~\eqref{eq:BdG_Gamma} in the semiclassical regime, that is on space-time scales where the Planck constant $h$ becomes negligible. To make connection with earlier studies on the effective approximation of many-body interacting fermionic systems (see Section~\ref{sec:application}), we set $N=\Tr{\gamma}$ and write
	\begin{equation*}
		U(x)=\frac{1}{N}\,K(x),
	\end{equation*}
	where $K:\Rd\to\R$ is independent of $N$ and $\hbar$ and satisfies Conditions~\eqref{conditions:interaction_potential_U}, and the factor $N^{-1}$ is the mean-field coupling constant.

	In the context of the semiclassical limit, it is convenient to define the rescaled operator 
	\begin{equation}\label{eq:scale-gamma}
		\op :=\frac{1}{Nh^d}\,\gamma,	
	\end{equation}
	so that $h^d \Tr{\op} = 1$. We will call positive operators verifying this trace normalization density operators and denote this class of operators by $\opP(\h)$. We also define the semiclassical Schatten norms, which are quantum analogues of the Lebesgue norms on the phase space, by
	\begin{equation}\label{eq:Schatten_rescaled}
		\Nrm{\op}{\L^p} := h^{d/p} \(\Tr{\n{\op}^p}\)^{1/p}.
	\end{equation}
	These norms are helpful in identifying the necessary scaling of quantum objects that will lead to a nontrivial semiclassical limit. 
	
	In the case of zero pairing, we see that the scaling~\eqref{eq:scale-gamma} leads to rewrite Equation~\eqref{eq:unscaled_Hartree-Fock} as follows 
	\begin{equation}\label{eq:Hartree--Fock}
		i\hbar\,\partial_t \op = \com{\sfH_{\op},\op}\quad \text{ with } \quad \sfH_{\op} = \frac{\n{\opp}^2}{2}+ V_{\op} - h^d\, \sfX_{\op}\,,
	\end{equation}
	where $V_{\op} = K*\varrho(x)$ and $\varrho(x)$ is the spatial distribution of particles defined by 
	\begin{equation}\label{eq:spatial_density_quantum}
		\varrho(x) = \Diag{\op}(x) := h^d\, \op(x,x)
	\end{equation}
	and the exchange operator $\sfX_{\op}$ has the integral kernel
	\begin{equation}\label{eq:exchange_operator}
		\sfX_{\op}(x,y) = K(x-y)\, \op(x,y)\,.
	\end{equation}
	Furthermore, with this scaling, it is known that in the semiclassical limit $\hbar\to 0$ one can recover classical phase space dynamics from the Hartree--Fock dynamics (see e.g. \cite{lions_sur_1993, benedikter_hartree_2016, lafleche_strong_2023}). More precisely, one obtains as the semiclassical approximation of Equation~\eqref{eq:Hartree--Fock} the Vlasov equation
	\begin{equation}\label{eq:Vlasov}
		\dpt f + \xi\cdot\Dchi f + E_f \cdot\Dv f = 0 \quad \text{ with } \quad f(0, \chi, \xi)= f^{\init}(\chi, \xi)\ge 0\,,
	\end{equation}
	where $f$ is a time-dependent probability density on $\Rdd = \Rd_\chi\times \Rd_\xi$, $E_f = -\nabla V_f$ is the self-consistent force field associated to the mean-field potential $V_f(\chi)=(K*\rho_f)({t,}\chi)$ with $\rho_f$ the spatial density defined by
	\begin{equation*}
		\rho_f({t,} \chi)=\intd f(t,\chi,\xi) \d\xi\,.
	\end{equation*}

	In this work, we want to extend the above result to the case of the BdG equation \eqref{eq:BdG}, with a non-zero pairing operator $\alpha$. To this end, we define the two-particle operator
	\begin{equation}\label{eq:op-alpha-sc}
		\op_{\alpha}=\lambda\ket{\alpha}\!\!\bra{\alpha}
	\end{equation}
	as the orthogonal projection in $\h\otimes\h$ onto the function defined by $(x,y)\mapsto\alpha(x,y)$, and where $\lambda>0$ is chosen so that $h^{2d} \Tr{\op_{\alpha}}=1$. This allows us to rewrite the BdG equation~\eqref{eq:BdG} in the equivalent form \eqref{eq:semiclassical_BdG}, as shown in Section~\ref{sec:rescaling}, and compare it with its classical analogue 
	\begin{subequations}\label{eq:Vlasov_F_error_1}
		\begin{align}\label{eq:Vlasov_and_error_1}
			&\dpt f + \xi\cdot\grad_{\chi} f + E_f\cdot\Dv f = 0
			\\\label{eq:F_and_error_1}
			&\dpt F + \xi_{12}\cdot\nabla_{\chi_{12}} F + E_{12}\cdot\nabla_{\xi_{12}} F = \frac{1}{N} \,\nabla K(\chi_1-\chi_2)\cdot \(\nabla_{\xi_1}-\nabla_{\xi_2}\)F
		\end{align}
	\end{subequations}
	where $F(z_{12}) = F(z_1,z_2)$ is the two-particle distribution function defined on $\Rdd\times\Rdd$, $\chi_{12}:=(\chi_1,\chi_2)\in\Rdd_\chi$, $\xi_{12}:=(\xi_1,\xi_2)\in\Rdd_\xi$ and $E_{12}:=(E_f(\chi_1),E_f(\chi_2))$. 
	
	Notice that System~\eqref{eq:Vlasov_F_error_1} is well-posed. Indeed, being that Equation~\eqref{eq:Vlasov_and_error_1} is a Vlasov equation with smooth interaction $K$, it is well-posed by standard techniques (see e.~g.~\cite{dobrushin_vlasov_1979}). Given $f$, a solution of equation~\eqref{eq:Vlasov_and_error_1}, and the corresponding vector field $E_{12}$, equation~\eqref{eq:F_and_error_1} is simply a linear transport equation with smooth vector field, which is well-posed by standard characteristics methods.

\subsection{Semiclassical optimal transport pseudo-metric}

	We now introduce the tools that we will use in our main theorem to prove the semiclassical limit of the BdG equation \eqref{eq:semiclassical_BdG}. 
	
	Denote $z = (\chi, \xi)\in \Rdd$. Let $f$ be a probability density function on $\Rdd$ and $\op\in \opP(\h)$. A coupling of $f$ and $\op$ is a measurable function $\opgam: z\mapsto \opgam(z)$ defined for almost all $z \in \Rdd$ with values in the space of bounded linear operators acting on $\h$ such that, for almost all $z \in \Rdd$, we have that $\opgam(z) \ge 0$ and it satisfies the conditions
	\begin{equation}
		h^d \tr_{\h}\!\(\opgam(z)\) = f(z)\quad \text{ and }\quad \intgam \opgam(z)\d z = \op \, .
	\end{equation}
	The set of all couplings of $f$ and $\op$ is denoted by $\cC(f, \op)$. Next, we define the semiclassical optimal transport pseudo-metric by
	\begin{equation}\label{eq:wasserstein-1}
		\Wh(f, \op) =\(\inf_{\opgam \in \cC(f, \op)}\intgam h^d \tr_\h\(\opc(z)\,\opgam(z)\)\d z\)^\frac12
	\end{equation} 
	with the cost function $\opc(z)$ defined by the unbounded operator whose action on test functions $\varphi$ gives
	\begin{equation*}
		(\opc(z)\varphi)(x) = \n{\chi-x}^2\varphi(x)+\n{\xi-\opp}^2\varphi(x)
	\end{equation*}
	and $\opp = -i\hbar\Dx$. The notation $\tr_\h(\opc\,\opgam)$ should be understood as $\tr_\h(\opc^{1/2}\,\opgam\,\opc^{1/2})$ in general if $\opc\,\opgam$ is not trace class. The above semiclassical optimal transport pseudo-metric between density operators and classical phase space functions was first introduced in~\cite{golse_schrodinger_2017}. It can be viewed as an intermediate notion between the classical Monge--Kantorovich distance (Wasserstein distance) of exponent 2 on the space of Borel probability measures and the quantum optimal transport pseudo-metric on the space of density operators defined in \cite{golse_mean_2016}. The properties of these pseudo-metrics can be found in~\cite{golse_schrodinger_2017,golse_optimal_2022, caglioti_towards_2022, lafleche_quantum_2023}.
 
	Likewise, for any two-particle probability density function $F(z_1, z_2)$ on $\Rdd\times\Rdd$ and two-particle density operator $\op_2\in \opP(\h\otimes\h)$, we denote by $\Wh(F,\op_2)$ their semiclassical optimal transport pseudo-metric, defined in the same way with $\h$ replaced by $\h\otimes\h$ and $\Rdd$ replaced by $\Rdd\times\Rdd$.
	
\subsection{Main results}
 
	We are now ready to state our main results. In our first theorem, we will be concerned with the limit to the following classical equations corresponding to Equations \eqref{eq:Vlasov_and_error_1}--\eqref{eq:F_and_error_1} with $N=\infty$.
	\begin{equation}\label{eq:easy_equations}
		\begin{aligned}
			&\dpt f + \xi\cdot\grad_{\chi} f + E_f\cdot\Dv f = 0\,,
			\\
			&\dpt F + \xi_{12}\cdot\nabla_{\chi_{12}} F + E_{12}\cdot\nabla_{\xi_{12}} F = 0\,.
		\end{aligned}
	\end{equation}

	\begin{theorem}\label{thm:main}
			Let $d\geq 3$ and assume $N\hbar\geq C$ where $C$ does not depend on $N$ and $\hbar$, $K$ be an even, real-valued function such that $\nabla^2K\in L^\infty(\Rd,\Rdd)$, $\widehat{K}\in L^1(\Rd)$ and $x\mapsto \n{x} K(x)\in L^\infty(\Rd)$. Let $(\gamma,\alpha)$ be a solution of the BdG equations \eqref{eq:BdG} and let $\op$ and $\op_\alpha$ be their rescaled versions defined in~\eqref{eq:scale-gamma} and \eqref{eq:op-alpha-sc} with initial data $(\op^\init, \op_\alpha^\init) \in \opP(\h)\times \opP(\h^{\otimes 2})$ such that 
			\begin{equation*}
				h^d\Tr{\op^\init\n{\opp}^4}, \qquad h^d\Tr{\op^\init\n{x}^4} \quad \text{ and } \quad \Nrm{\op^{\init}}{\L^d}
			\end{equation*}
			are uniformly bounded in $\hbar$. Let $(f, F)$ be the solutions of the system~\eqref{eq:easy_equations} with initial conditions $(f^{\init}, F^{\init})$, which are probability density functions defined on $\Rdd$ and $\Rdd\times\Rdd$, respectively, and such that 
		\begin{equation*}
			\intgam \n{z}^2f^\init(z)\d z<\infty \quad \text{ and } \quad \intggam \(\n{z_{1}}^2+\n{z_2}^2\)\! F^\init(z_{1}, z_{2})\d z_1\d z_2<\infty\,.
		\end{equation*}
		Then there exist a constant $C$ dependent on $\Nrm{K}{C^2}$ and the semiclassical Schatten norms of $\op^\init$ but independent of $N$ and $\hbar$, such that for any $t\geq 0$
		\begin{equation}\label{eq:main}
			\Wh(f,\op)^2 + \Wh(F,\op_\alpha)^2 \le \(\Wh(f^\init, \op^\init)^2+\Wh(F^\init, \op_\alpha^\init)^2 + \hbar\) e^{C \,e^{C\,t}} .
		\end{equation}
	\end{theorem}
	
	\begin{remark}\label{remark:opt-transp-distance}
		The optimal transport pseudo-metrics that are used in the above theorem are not distances since $\Wh(f,\op) \geq d\,\hbar$ (see~\cite{golse_optimal_2022}). However, they still imply convergence in the semiclassical regime $\hbar\to 0$. More precisely, introducing the Wigner transform
		\begin{equation}\label{eq:wigner}
			f_{\op}(\chi,\xi) = \intd e^{-i\xi\cdot y/\hbar}\,\op(\chi+\tfrac{y}{2}, \chi-\tfrac{y}{2})\d y\, ,
		\end{equation}
		and the Husimi transform $\tilde{f}_{\op} = g_h * f_{\op}$ with $g_h(z) = \(\pi\hbar\)^{-d}\,e^{-\n{z}^2/\hbar}$, it holds (see~\cite[Theorem~2.4]{golse_schrodinger_2017})
		\begin{equation*} 
			\W_2(f,\tilde{f}_{\op})^2 \leq \Wh(f,\op)^2 + d\,\hbar
		\end{equation*}
		and so convergence of $\Wh(f,\op)$ to $0$ implies the convergence of the Husimi transform of $\op$ (and so also its Wigner transform, see~\cite{lions_sur_1993}) to $f$. On the other hand, the right-hand side of Inequality~\eqref{eq:main} is also small initially and, if the operator $\op$ is sufficiently regular, $\tilde{f}_{\op}$ is close to $f$, as follows from the following inequality which follows from~\cite[Theorem~1.1]{lafleche_quantum_2023} and~\cite[Theorem~3.5]{golse_semiclassical_2021}
		 \begin{equation*}
 			\Wh(f,\op) \leq \W_2(f,\tilde{f}_{\op}) + \sqrt{d\hbar} + D_{\op} \,\hbar
 		\end{equation*}
 		where $D_{\op} = \|\nabla f_{\sqrt{\op}}\|_{L^2(\Rdd)}$ is proportional to the Wigner--Yanase skew information of $\op$.
	\end{remark}
	
	\begin{remark}
		The double-exponential growth on the right-hand side of~\eqref{eq:main} is due to the propagation of the Schatten norm $\L^d$ for $\op$. We can get a better bound in terms of time dependence in the regime where $Nh^d$ is of order $1$. Indeed,  in this case $\Nrm{\op}{\L^d} \leq 1$ and then the $e^{C\,e^{Ct}}$ can be replaced by a function of the form $e^{\Lambda(t)}$ for some polynomial function $\Lambda$. 
	\end{remark}
	
	In this next theorem, we consider semiclassical regime which allows us to obtain a nontrivial order $1/N$ two-body interaction effect correction to the dynamics of $F$.
	\begin{theorem}\label{thm:main_2}
		Under the same assumptions as in Theorem~\ref{thm:main} but with $Nh\to 0$ and $(f, F)$ solutions of the system~\eqref{eq:Vlasov_F_error_1}, let  $\Nrm{\alpha^\init}{L^2}\le C N h$. Then there exists $T$ and $C_T$, independent of $\hbar$ and $N$ but dependent on $\Nrm{\nabla^2 K}{L^\infty}$ and the semiclassical Schatten norms of $\op^\init$, such that for any $t\in[0,T]$,
		\begin{equation}\label{eq:main_2}
			\Wh(f,\op)^2 + \Wh(F,\op_\alpha)^2 \le C_T \(\Wh(f^\init, \op^\init)^2+\Wh(F^\init, \op_\alpha^\init)^2 + \hbar\).
		\end{equation}
	\end{theorem}
 
	\begin{remark}
		As an immediate consequence of Theorem \ref{thm:main} and the main result of Marcantoni et al. in \cite[Theorem 3.3]{marcantoni_dynamics_2023}, we establish a global-in-time joint semiclassical and mean-field approximation of the dynamics of a system of spin-$\frac12$ fermions with quasi-free initial data that are close to Slater determinant-like states by solutions of the Vlasov equation. In particular, we establish the convergence in some negative Sobolev space. See Theorem~\ref{thm:joint_limit} in Section~\ref{sec:application}.
	\end{remark}

\subsection{Previously known results}

	As the BdG equation \eqref{eq:BdG} can be seen as a generalization of the Hartree--Fock equation, we briefly review the literature concerning the semiclassical limit from the Hartree--Fock equation to the Vlasov equation. Equation~\eqref{eq:Vlasov} can be seen as the semiclassical approximation of a system of many interacting quantum particles, as pointed out in the pioneering works by Narnhofer and Sewell~\cite{narnhofer_vlasov_1981} and by Spohn~\cite{spohn_vlasov_1981} where the Vlasov equation was obtained directly from the many-body Schr\"odinger equation with smooth interaction in the combined mean-field and semiclassical regime. This has been reconsidered in~\cite{graffi_mean-field_2003} and more recently in~\cite{chong_many-body_2021, chen_combined_2021}, where the case of the Coulomb potential with a $N$ dependent cut-off has been addressed. Moreover, a combined mean-field and semiclassical limit for particles interacting via the Coulomb potential has been treated in~\cite{golse_mean-field_2021} for factorized initial data whose first marginal is given by a monokinetic Wigner measure (that can be seen as the Klimontovich solutions to the Vlasov equation), which leads to the pressureless Euler--Poisson system.

	Most of the above mentioned works rely on compactness methods that do not allow for an explicit bound on the rate of convergence, which is essential for applications. For this reason the Hartree equation~\eqref{eq:Hartree--Fock} has been considered as an intermediate step to decouple the problem into two separate parts, namely to prove the convergence of the mean-field limit from the many-body Schr\"odinger equation towards the Hartree equation, and then the semiclassical limit from the Hartree equation to the Vlasov equation. In this paper, we are interested in the latter problem, that has been largely studied in different settings. It was first proven by Lions and Paul in~\cite{lions_sur_1993}, and later in~\cite{markowich_classical_1993, figalli_semiclassical_2012}, that the Wigner transforms of the solutions of the Hartree equation~\eqref{eq:Hartree--Fock} converge in some weak sense to solutions of the Vlasov--Poisson equation. Quantitative rates of convergence were then obtained, first in the case when the Coulomb potential is replaced by a smoother potential, in Lebesgue-type norms~\cite{athanassoulis_strong_2011-1, athanassoulis_strong_2011, amour_semiclassical_2013, benedikter_hartree_2016} and in a quantum analogue of the Wasserstein distances \cite{golse_schrodinger_2017}. The case of singular interactions was then treated in~\cite{lafleche_propagation_2019, lafleche_global_2021} with the same quantum Wasserstein distances, and in~\cite{saffirio_semiclassical_2019, saffirio_hartree_2020, lafleche_strong_2023} in Lebesgue-type norms. In particular, for $K=\n{x}^{-1}$, the explicit rate has been established in~\cite{lafleche_propagation_2019, iacobelli_enhanced_2024} for the weak topology and in~\cite{saffirio_hartree_2020, lafleche_strong_2023} for the Schatten norms.

	In a different setting, the semiclassical limit has also been studied for local perturbations of stationary states in the case of infinite gases in~\cite{lewin_hartree_2020}.

	The BdG equation is known to offer a self-consistent field description of a system of fermionic particles (See \cite{de_gennes_superconductivity_1999}). The global well-posedness in the energy space of the time-dependent BdG equation in $\R^3$, with potential $U$ including the Coulomb potential and $\hbar$ fixed, can be found in~\cite{benedikter_diracfrenkel_2018}. This result was subsequently improved in~\cite{dong_dynamics_2021} to include positive singular potentials up to and including $U(x)=\n{x}^{-2+\epsilon}$, for $0<\epsilon\le 2$, via techniques from dispersive PDE theory. In fact, well-posedness and finite-time blowup of solutions to the BdG equation in energy space with a pseudo-relativistic kinetic energy were discussed prior in~\cite{hainzl_blowup_2010, lenzmann_minimizers_2010}. For completeness, let us also mention the fact that the well-posedness theory of a related system of coupled equations, also called the time-dependent Hartree--Fock--Bogoliubov equations in the ``spinless bosonic'' setting was first studied locally in time in~\cite{grillakis_beyond_2013, grillakis_pair_2017, grillakis_uniform_2019} for the pure state case and improved to global-in-time results along with obtaining global-in-time dispersive estimates in~\cite{chong_global_2021, chong_dynamical_2022,huang_global_2022}. The equations were also studied in the mixed state case in~\cite{bach_time-dependent_2022}.
	
	It is also worth mentioning that the Hartree--Fock--Bogoliubov equations were recently obtained in \cite{marcantoni_dynamics_2023} as the mean-field approximation of a system of $N$ interacting fermions with initial state close to quasi-free states with non-zero pairing operator. For the associated equilibrium problem, namely the study of the Hartree--Fock--Bogoliubov functional and its connection to BCS theory of superconductivity and superfluidity, we refer to the review papers \cite{bach_generalized_1994, hainzl_bardeencooperschrieffer_2016} and references therein.

\subsection{Plan of the paper}

	The rest of the paper is organized as follows. In Section~\ref{sec:strategy} we present the outline of the proof, give a useful equivalent formulation of the BdG equation and present some preliminary estimates. Section~\ref{sec:proof} is devoted to the proof of the main results, while Section~\ref{sec:application} provides an application of Theorem~\ref{thm:main} in the setting of the work \cite{marcantoni_dynamics_2023} about mean-field theory for interacting fermionic systems with non-zero pairing.

\section{The strategy: semiclassical Bogoliubov--de Gennes equation}\label{sec:strategy}

	In this section, we present the strategy of the proof, that relies on an ad hoc rewriting of the BdG equation~\eqref{eq:BdG} in the form \eqref{eq:semiclassical_BdG}, representing the main novelty of our approach.

	We first recall that the case of the zero pairing relies on the correspondence principle. It is indeed easy to see that the Vlasov equation~\eqref{eq:Vlasov} can be written in terms of Poisson brackets as $\partial_t f=\{H_f,\,f\}$, where $H_f=\frac{|\xi|^2}{2}+V_f$ and the braces denote the Poisson brackets. Then, the correspondence principle of quantum mechanics, together with the observation that the exchange term \eqref{eq:exchange_operator} vanishes as $h\to 0$, implies that the Hartree--Fock evolution \eqref{eq:Hartree--Fock} converges to the Vlasov dynamics \eqref{eq:Vlasov} for $h$ small.	
	
	In the case of non-zero pairing, the semiclassical approximation of the BdG equation is less clear. The main difficulty comes from the fact that the correspondence principle of quantum mechanics is not immediately applicable to the pairing operator. Our strategy consists in recasting the problem for $\gamma$ and $\alpha$ in terms of the positive self-adjoint density operators $\op$ and $\op_\alpha$ and consider their time evolution.

	\subsection{Scaling the BdG equation}\label{sec:rescaling} 
	To study the semiclassical limit of the pairing operator, we start by noticing from Conditions~\eqref{eq:eq_alpha_0} that it follows
	\begin{equation}\label{eq:theta_expression}
		\theta_\alpha := \frac{1}{N} \Tr{\n{\alpha}^2} \le 1 - Nh^d\Nrm{\op}{\L^2}^2 \in [0,1)\,.
	\end{equation} To better understand $\alpha$, it is more natural to consider its integral kernel and view the kernel as a two-particle wave function. Hence, assuming $\alpha\neq 0$, we define the normalized pairing wave function as
	\begin{equation*}
		\Psi_\alpha(x_1,x_2) := \frac{1}{\Nrm{\alpha}{L^2}} \,\alpha(x_1,x_2)\, .
	\end{equation*}
	Following our scaling convention of $\gamma$ and Identity \eqref{eq:eq_alpha_0}, it is suggestive to consider the rescaled projection operator acting on $\h\otimes\h$ and its normalization 
	\begin{align}
		\sfA_{\alpha}:=\frac{1}{N\,h^{2d}}\ket{\alpha}\!\!\bra{\alpha}\quad \text{ and } \quad \op_\alpha := h^{-2d} \ket{\Psi_\alpha}\!\!\bra{\Psi_\alpha}.
	\end{align}
	Clearly, $\sfA_\alpha = \theta_\alpha \op_\alpha$ and $\op_\alpha$ satisfies the normalization $h^{2d} \Tr{\op_\alpha} = 1$. 
	
	To make connection with classical phase space dynamics, we need to recast the BdG equation in terms of the rescaled operators $\op$ and $\op_{\alpha}$. Define $\sfX_{\alpha}$ by Expression \eqref{eq:exchange_operator} and notice that $\sfX_\alpha\,\alpha^* = Nh^{2d}\theta_\alpha\tr_2\!\(K_{12}\op_\alpha\)$, which then implies 
	\begin{equation}\label{eq:skew_symmetry_of_Xalpha_alpha}
		\sfX_\alpha\,\alpha^* - \alpha\,\sfX_\alpha^* = Nh^{2d}\theta_\alpha\tr_2\!\(\com{K_{12},\op_\alpha}\) =: Nh^d\theta_\alpha \com{K_{12},\op_\alpha}_{:1}.
	\end{equation} 
	Here, $\tr_2(\cdot)$ denotes the partial trace with respect to the second Hilbert space and $K_{12}$ denotes the operator of multiplication by $K(x_1-x_2)$ on $\h\otimes\h$. 
	Hence, we see that Equation~\eqref{eq:BdG_gamma} has the form 
	\begin{equation*}
		i\hbar\,\dpt \op = \com{\sfH_{\op},\op} + \theta_\alpha\com{\tfrac{1}{N}K_{12},\op_\alpha}_{:1}
	\end{equation*}
	where $\sfH_{\op}$ is as defined in Equation \eqref{eq:Hartree--Fock}.
	
	To rewrite Equation \eqref{eq:spinless_BdG_alpha}, we view $\alpha = \alpha(x_1,x_2)$ as a two-body wave function as opposed to it being a Hilbert--Schmidt operator. Then the equation has the form 
	\begin{equation}
		\begin{aligned}
			i\hbar\,\dpt \alpha(x_1, x_2) &= \(-\tfrac{\hbar^2}{2}\lapl_{x_1} - \tfrac{\hbar^2}{2}\lapl_{x_2} + \tfrac{1}{N}K(x_1-x_2)\)\alpha(x_1, x_2)
			\\
			&\quad + h^d\intd \(K(x_1-y)+K(y-x_2)\)\op(y, y)\,\alpha(x_1, x_2)\d y
			\\
			&\quad - h^d\intd \(K(x_1-y)+K(y-x_2)\)\op(x_1, y)\,\alpha(y, x_2)\d y
			\\
			&\quad - h^d\intd \(K(x_1-y)+K(y-x_2)\)\op(x_2, y)\,\alpha(x_1, y)\d y
		\end{aligned}
	\end{equation}
	or more compactly
	\begin{equation}\label{eq:rescaled_BdG_alpha}
		i\hbar\, \dpt \alpha = \(\sfH_{12}+\tfrac{1}{N}\,K_{12} \)\alpha - h^d\,\op_{12}\, K_{12}\,\alpha
	\end{equation}
	where $\op_{12} := \op\otimes \id + \id\otimes \op$ and $\sfH_{12} := \sfH_{\op}\otimes \id + \id\otimes \sfH_{\op}$. 
	
	To make connection with classical mechanics, it is better to consider the two-particle operator $\sfA_{\alpha}$ instead of $\alpha$. Using Equation \eqref{eq:rescaled_BdG_alpha}, it follows that $\sfA_\alpha$ satisfies
	\begin{equation*}
		i\hbar\,\dpt \sfA_\alpha = \com{\sfH_{12} + \tfrac{1}{N}K_{12}\(\id-Nh^d\op_{12}\), \sfA_\alpha} - h^d \com{K_{12}, \op_{12}}\sfA_\alpha\,.
	\end{equation*}
	Since $\theta_\alpha = h^{2d}\Tr{\sfA_\alpha}$, taking the trace of the above equation yields
	\begin{equation}\label{eq:dt_theta}
		i\hbar\,\dt\theta_\alpha 
		= -h^d \bangle{\com{K_{12},\op_{12}}}_{\sfA_\alpha}
	\end{equation}
	where $\bangle{B}_A := h^{2d}\Tr{AB}$ if $A$ and $B$ are operators acting on $L^2(\Rdd)$. Finally, summarizing the above discussion and using the fact that $\sfA_\alpha=\theta_\alpha\op_\alpha$, we obtain the equations%
	\begin{subequations}\label{eq:semiclassical_BdG}%
		\begin{align}\label{eq:semiclassical_gamma}
			i\hbar\,\dpt \op &= \com{\sfH_{\op},\op} + \frac{\theta_\alpha}{N} \com{K_{12},\op_\alpha}_{:1} 
			\\\label{eq:semiclassical_alpha}
			i\hbar\,\dpt \op_\alpha &= \com{\sfH_{12} + \tfrac{1}{N}K_{12}\(\id-Nh^d\op_{12}\), \op_\alpha} 
			\\
			&\quad + h^d \(\com{K_{12}, \op_{12}} - \bangle{\com{K_{12}, \op_{12}}}_{\op_{\alpha}}\)\op_\alpha\,. \notag
		\end{align}
	\end{subequations}
	In particular, the trace of $\op$ and $\op_\alpha$ are conserved.
	
	Now, by the correspondence principle, one can expect, at least in the case when $K$ is a sufficiently regular potential, that $F(z_{12}) = F(z_1,z_2) := F_{\op_\alpha}(z_1,z_2)$ defined on $\Rdd\times\Rdd$ solves in the limit $\hbar\to 0$ and $N\to\infty$
	\begin{equation}
		\dpt F + \xi_{12}\cdot\nabla_{\chi_{12}} F + E_{12}\cdot\nabla_{\xi_{12}} F = 0
	\end{equation}
	where $\xi_{12}=(\xi_1, \xi_2)$ and $E_{12} := \(E_f(\chi_1), E_f(\chi_2)\)$ with 
	\begin{align*}
		E_f(\chi):= -\grad V_{f} \quad \text{ and } \quad V_f(\chi)= \iintd K(\chi-\chi_1)f(z_1)\d z_1\,.
	\end{align*}
	In fact, if $Nh^d \ll 1$, then, to the order $1/N$, we expect to have
	\begin{align*}
		i\hbar\,\dpt \op &= \com{\sfH_{\op},\op} + \tfrac{\theta_\alpha}{N}\com{K_{12},\op_\alpha}_{:1}
		\\
		i\hbar\,\dpt \op_\alpha &= \com{\sfH_{12} + \tfrac{1}{N}K_{12}, \op_\alpha}
	\end{align*}
	as the leading order dynamics, that is formally, when $\hbar \to 0$ we have that
	\begin{subequations}\label{eq:Vlasov_and_error_2}
		\begin{align}
			&\dpt f + \xi\cdot\grad_{\chi} f + E_f\cdot\Dv f = \frac{\theta_\alpha}{N} \intgam 
			\nabla K(\chi-\chi_2) \cdot \Dv F(z,z_2) \d z_2\,,
			\\
			&\dpt F + \xi_{12}\cdot\nabla_{\chi_{12}} F + E_{12}\cdot\nabla_{\xi_{12}} F = \frac{1}{N} \nabla K(\chi_1-\chi_2)\cdot \(\nabla_{\xi_1}-\nabla_{\xi_2}\)F\,.
		\end{align}
	\end{subequations}

	\begin{remark}
		As already pointed out in Remark~\ref{remark:opt-transp-distance}, estimates in the semiclassical optimal transport distance give accuracy up to $\hbar$ since $\Wh(f,\op)^2 \geq d\hbar$. In particular, the terms on the right hand side of Equations \eqref{eq:Vlasov_and_error_2} with the $1/N$ in front are meaningful only if $h \ll 1/N$, which is allowed from some semiclassical regime but does not include for instance the regime $N\, h^d = 1$.
	\end{remark}
	
	\subsection{Conservation laws and a priori estimates}

	As indicated in the previous discussion, if $(\op, \op_\alpha)$ solves System \eqref{eq:semiclassical_BdG}, then it follows that
	\begin{equation}
		h^d\tr_{\h}\!\(\op\) = 1 \quad \text{ and } \quad h^{2d}\tr_{\h\otimes\h}\!\(\op_{\alpha}\)=1
	\end{equation}
	hold for all $t \ge 0$ provided the equalities hold at initial time. Define the one-particle density operator (first marginal) associated to $\op_\alpha$ by
	\begin{equation}\label{eq:op_alpha:1}
		\op_{\alpha:1} := h^d \tr_2(\op_\alpha) = \frac{1}{Nh^d\theta_\alpha} \n{\alpha^*}^2\ge 0\, .
	\end{equation}
	In light of the semiclassical scaling, the last inequality in Formula~\eqref{eq:eq_alpha_0} gives
	\begin{equation}\label{eq:rescaled_quasifree_id}
		N\,h^d\op^2 + \theta_\alpha\, \op_{\alpha:1} \leq \op\, ,
	\end{equation} 
	which is preserved by the BdG dynamics. As an immediate consequence, we have
	\begin{equation}\label{eq:ineq_alpha_2}
		0\leq \theta_\alpha\,\op_{\alpha:1} \leq \op \leq \frac{1}{Nh^d}\,\id\, .
	\end{equation}
	Moreover, since $\op_\alpha = h^{-2d} \ket{\Psi_\alpha}\!\!\bra{\Psi_\alpha}$ is a rank one operator, it verifies $0 \leq \op_{\alpha} \leq h^{-2d}\id$. The right-hand side inequality is sharp and more generally $\Nrm{\op_{\alpha}}{\L^p} = h^{-2d/p'}$.

	As proved for instance in~\cite{benedikter_diracfrenkel_2018}, the following energy functional is conserved
	\begin{equation}
		\Eps := h^d\Tr{\n{\opp}^2\op} + \frac{1}{2} \intd V_{\op}\, \rho - \frac{h^{2d}}{2} \Tr{\sfX_{\op}\,\op} + \frac{\theta_\alpha}{N}\,h^{2d}\Tr{K_{12}\,\op_\alpha}. 
	\end{equation}
	In particular, if $K\in L^\infty$ and the energy is initially bounded uniformly in $\hbar$, then the kinetic energy of $\op$ is bounded uniformly in $\hbar$ and time and more precisely
	\begin{equation}\label{est:kinetic_energy_bound}
		h^d\Tr{\n{\opp}^2\op} \leq \Eps + \(1 + \tfrac{\theta_\alpha}{N}\) \Nrm{K}{L^\infty} \leq \Eps + 2 \Nrm{K}{L^\infty} =: C_{\Eps,K}\,.
	\end{equation}
	Moments of order $2$ of $\op_{\alpha:1}$ are also bounded uniformly in $\hbar$ and time by the energy since by Formula~\eqref{eq:ineq_alpha_2}
	\begin{equation*}
		\theta_\alpha\,h^d\Tr{\op_{\alpha:1} \n{\opp}^2} \leq h^d\Tr{\op \n{\opp}^2} \leq C_{\Eps,K}\,.
	\end{equation*}
	For the two-particles density operator $\op_{\alpha}$, this can be written $h^{2d}\Tr{\op_\alpha\n{\opp_1}^2} \leq C_{\Eps,K}$ where $\opp_1$ is the momentum operator acting on the first variable. By symmetry, the same is true by replacing $\opp_1$ by $\opp_2$ and so we deduce that
	\begin{equation*}
		\theta_\alpha h^{2d}\Tr{\op_\alpha\(\n{\opp_1}^2+\n{\opp_2}^2\)} \leq C_{\Eps,K}\,.
	\end{equation*}	
	We can propagate higher order moments. In our case, it will be sufficient to propagate order $4$ moments, as shown in the following proposition.
	\begin{prop}\label{prop:semiclassical_propagation_of_moments}
		Let $(\op,\op_\alpha)$ be a solution of the BdG equation~\eqref{eq:semiclassical_BdG} and
		\begin{equation*}
			M_n := h^d\Tr{\op\n{\opp}^n} \quad \text{ and } \quad N_n := h^d\Tr{\op\n{x}^n}
		\end{equation*}
		denote the velocity and position moments  of order $n\in\N$ of the operator $\op$. Then for any $t\geq 0$,
		\begin{align*}
			M_2(t) &\leq C_{\Eps,K} & M_4(t)^{1/2} &\leq M_4(0)^{1/2} + C_K\,t
			\\
			N_2(t)^{1/2} &\leq N_2(0)^{1/2} + C_{\Eps,K}^{1/2}\,t & N_4(t)^{1/4} &\leq N_4(0)^{1/4} + C\,\hbar^3 t + C\(M_4(0)^{1/2} + t\)^\frac{3}{2}
		\end{align*}
		where $C_K = 3 \(\hbar \Nrm{\Delta K}{L^\infty} + 2 \Nrm{\nabla K}{L^\infty} \sqrt{C_{\Eps,K}}\)$ and $C$ only depends on $d$ and $C_K$.
	\end{prop}
	
	\begin{proof}
		To simplify the computations, we write the evolution equation for $\op$ given by Equation~\eqref{eq:semiclassical_gamma} in the form
		\begin{equation*}
			i\hbar \,\dpt \op = \frac{1}{2}\com{\n{\opp}^2,\op} + \com{K_{12},\opmu}_{:1} \quad \text{ with } \quad \opmu = \op^{\otimes 2}\(\id-\sfX_{12}\) + \tfrac{\theta_\alpha}{N}\op_\alpha
		\end{equation*}
		where $\sfX_{12}$ is the operator that exchanges the first and the second coordinate, that is $\sfX_{12} \varphi(x_1,x_2) = \varphi(x_2,x_1)$. Observe that $\opmu$ is self-adjoint. Then it follows from the cyclicity of the trace that
		\begin{equation*}
			i\hbar\,\dt M_4 = h^{2d} \tr_{\h^{\otimes 2}}\!\(\com{K_{12},\opmu}\n{\opp_1}^4\) = h^{2d} \tr_{\h^{\otimes 2}}\!\(\opmu \com{\n{\opp_1}^4,K_{12}}\).
		\end{equation*}
		By the Leibniz formula for commutators, $\com{A^2,B} = A \com{A,B} + \com{A,B}A$, and since $\opmu$ is self-adjoint, it gives
		\begin{equation*}
			\hbar \,\dt M_4 = 2\,h^{2d} \im \tr_{\h^{\otimes 2}}\!\(\opmu \n{\opp_1}^2 \com{\n{\opp_1}^2,K_{12}}\).
		\end{equation*}
		Since $\com{\opp_1,K_{12}} = -i\hbar\,\nabla K_{12}$, where $\nabla K_{12}$ denotes the operator of multiplication by $\nabla K(x_1-x_2)$, it follows that
		\begin{equation*}
			\com{\n{\opp_1}^2,K_{12}} = -i\hbar\(\opp_1\cdot \nabla K_{12} + \nabla K_{12}\cdot \opp_1\) = -\hbar^2\,\Delta K_{12}  -2\,i\hbar\,\nabla K_{12}\cdot \opp_1
		\end{equation*}
		and so it follows from the cyclicity and H\"older's inequality for the trace that
		\begin{equation}\label{eq:dt_M4}
			\dt M_4 \leq 2\, \hbar \Nrm{\Delta K}{L^\infty} \Nrm{\opmu \n{\opp_1}^2}{\L^1(\h^{\otimes 2})} + 4 \Nrm{\nabla K}{L^\infty} \Nrm{\opp_1\,\opmu \n{\opp_1}^2}{\L^1(\h^{\otimes 2})}
		\end{equation}
		Now we decompose $\opmu$ into the three terms that define it and use the triangle inequality for the trace norm. Notice indeed that for the first term, we get
		\begin{align*}
			&\Nrm{\op^{\otimes 2} \n{\opp_1}^2}{\L^1(\h^{\otimes 2})} = \Nrm{\op \n{\opp}^2}{\L^1} \leq \Nrm{\sqrt{\op}}{\L^2} \Nrm{\sqrt{\op} \n{\opp}^2}{\L^2} = M_4^{1/2}
			\\
			&\Nrm{\op^{\otimes 2}\,\sfX_{12} \n{\opp_1}^2}{\L^1(\h^{\otimes 2})} = \Nrm{\op^{\otimes 2} \n{\opp_2}^2}{\L^1(\h^{\otimes 2})} = \Nrm{\op \n{\opp}^2}{\L^1} \leq M_4^{1/2}
			\\
			&\Nrm{\op_\alpha \n{\opp_1}^2}{\L^1(\h^{\otimes 2})} \leq \Nrm{\sqrt{\op_\alpha}}{\L^2} \Nrm{\sqrt{\op_\alpha} \n{\opp_1}^2}{\L^2} = h^d\big(\!\tr_{\h^{\otimes 2}}\!\big(\op_\alpha\n{\opp_1}^4\big)\big)^\frac{1}{2} \leq \theta_\alpha^{-\frac{1}{2}}\,M_4^{1/2}
		\end{align*}
		where the last inequality follows from Inequality~\eqref{eq:ineq_alpha_2}. Similarly, for the second term of the right-hand side of Inequality~\eqref{eq:dt_M4} use the fact that for $\opnu = \op^{\otimes 2}$, $\opnu = \op^{\otimes 2}\sfX_{12}$ or $\opnu = \theta_\alpha\,\op_\alpha$
		\begin{equation*}
			\Nrm{\opp_1\,\opnu \n{\opp_1}^2}{\L^1(\h^{\otimes 2})} \leq \Nrm{\opp_1\,\opnu}{\L^2(\h^{\otimes 2})} \Nrm{\opnu \n{\opp_1}^2}{\L^2(\h^{\otimes 2})} \leq M_2^{1/2}\,M_4^{1/2}
		\end{equation*}
		and this gives finally, since $N\geq 1$, $\theta_\alpha\leq 1$ and $M_2 \leq C_{\Eps,K}$,
		\begin{equation*}
			\dt M_4 \leq 6 \(\hbar \Nrm{\Delta K}{L^\infty} + 2 \Nrm{\nabla K}{L^\infty} \sqrt{C_{\Eps,K}}\) \sqrt{M_4}
		\end{equation*}
		from which the result follows by Gr\"onwall's Lemma.
		
		The propagation of position moments follows just by writing for $n=2$ or $n=4$
		\begin{equation*}
			i\hbar\,\dt N_n = \frac{1}{2} \Tr{\com{\n{\opp}^2,\op} \n{x}^n} = \frac{1}{2} \Tr{\op \com{\n{x}^n,\n{\opp}^2}}.
		\end{equation*}
		Therefore, since $\frac{1}{i\hbar} \com{\n{x}^n,\n{\opp}^2} = 2\(x\cdot\opp + \opp\cdot x\)$, it follows from the H\"older's inequality for Schatten norms that
		\begin{equation*}
			\dt N_2 = 2 \re\Tr{\op \,x\cdot\opp} \leq 2 \,M_2^{1/2} N_2^{1/2}
		\end{equation*}
		which yields the inequality for $N_2$ by Gr\"onwall's Lemma. On the other hand, it follows from \cite[Lemma~3.2]{lafleche_global_2021} that
		\begin{equation*}
			\dt N_4 = 2 \re\Tr{\op \n{x}^2 \(x\cdot\opp+\opp\cdot x\)} \leq C\(M_4^{1/4}\,N_4^{3/4} + \hbar\,N_2\).
		\end{equation*}
		By H\"older's inequality for Schatten norm, the fact that $N_0=1$ and Young's inequality for the product, $\hbar\,N_2 \leq \hbar\,N_4^{1/2} \leq \hbar^3 + N_4^{3/4}$, it gives a differential inequality for $y(t)=N_4(t)$ of the form $y' \leq C\(M_4^{1/4}+ \hbar^3\)y^{3/4}$, which again leads to the result by Gr\"onwall's Lemma.
	\end{proof}

	In the remainder of the section, we obtain uniform-in-$\hbar$ estimate for the semiclassical Schatten norms for $\op$ along the BdG dynamics in the case of bounded potential $K$ for different semiclassical scaling regimes.
	\begin{prop}\label{prop:schatten_propag}
		Let $\widehat{K}\in L^1$. Suppose $\op = \op(t)$ is a solution to Equation \eqref{eq:semiclassical_gamma} with $\op(0)=\op^\init \in \cL^p$. We have the following.
		\begin{enumerate}[(i)]
			\item\label{part:Nh>1_regime} In a regime where $N\hbar \geq C$ holds for some fixed $C>0$, independent of $N$ and $\hbar$, then there exists $C_K>0$, dependent only on $K$, such that we have the estimate 
			\begin{equation*}
				\Nrm{\op}{\L^p} \leq \Nrm{\op^\init}{\L^p} e^{C_K t}.
			\end{equation*}
			\item\label{part:N2h>1_regime} If $\theta_\alpha^\init \leq C \,N\,h^{2d/p}$ for some constant $C$ independent of $\hbar$, then there exists $C>0$ independent of $\hbar$ such that for any $t\in[0,T]$ with $T = C\,h^{1-d/p}$,
			\begin{equation}\label{eq:schatten_propag_2}
				\Nrm{\op}{\L^p} \leq C \quad \text{ and }\quad \theta_\alpha \leq C \,N\,h^{2d/p}.
			\end{equation}
		\end{enumerate}
	\end{prop}

	\begin{proof}
		By Equation \eqref{eq:semiclassical_gamma}, we have 
		\begin{align}\label{eq:propagation_schatten_p_of_op}
			\frac{1}{p}\,\frac{\d}{\d t} \Nrm{\op}{\L^p}^p &=
			\frac{h^d}{p}\,\frac{\d}{\d t} \Tr{\op^p} = \frac{2}{N^2\,\hbar} \im{\Tr{\op^{p-1}\, \sfX_\alpha\,\alpha^*}}.
		\end{align}
		Therefore, by H\"older's inequality, we obtain the bound
		\begin{equation*}
			\frac{\d}{\d t} \Nrm{\op}{\L^p} \leq \frac{4\pi}{N^2\,h^{d+1}} \Nrm{\sfX_\alpha}{\L^{2p}}\Nrm{\alpha^*}{\L^{2p}}.
		\end{equation*}
		Now it follows from the formula
		\begin{equation*}
			\sfX_\alpha = \intd \widehat{K}(\omega)\, e_\omega \,\alpha\, e_{-\omega} \d \omega
		\end{equation*}
		where $e_\omega$ is the operator of multiplication by the function $e_\omega(x) = e^{-2i\pi\, \omega\cdot x}$ that
		\begin{equation}\label{eq:Xop}
			\Nrm{\sfX_\alpha}{\L^{2p}} \leq C_K \Nrm{\alpha}{\L^{2p}},
		\end{equation}
		where $C_K = \snrm{\widehat K}_{L^1}$. Therefore, by Definition~\eqref{eq:op_alpha:1}, and Inequality~\eqref{eq:ineq_alpha_2}, we get
		\begin{equation}\label{est:derivative_schatten_p_op}
			\dt \Nrm{\op}{\L^p} \leq \frac{2\, C_K}{N\,\hbar} \Nrm{\theta_\alpha\op_{\alpha:1}}{\L^{p}}\le \frac{2\,C_K}{N\,\hbar} \Nrm{\op}{\cL^p}
		\end{equation}
		and Part~\ref{part:Nh>1_regime} follows from Gr\"onwall's lemma.

		To prove Part~\ref{part:N2h>1_regime}, we will also need to study the size of $\theta_\alpha$ along the BdG dynamics. By Equation~\eqref{eq:dt_theta}, cyclicity of the trace, symmetry of $\op_\alpha$ and $K_{12}$, and Equation~\eqref{eq:skew_symmetry_of_Xalpha_alpha}, we have 
		\begin{equation*}
			\dt \theta_\alpha 
			= \frac{8\pi h^{d-1}}{N} \im\Tr{\op\,\sfX_\alpha\,\alpha^*}.
		\end{equation*}
		Again, using H\"older's inequality and Inequality~\eqref{eq:Xop}, we get that 
		\begin{equation}\label{est:derivative_theta_alpha}
			\dt \theta_\alpha \leq \frac{4\, C_K}{N\,\hbar} \Nrm{\op}{\L^p} \Nrm{\alpha}{\L^{2p'}}^2 \leq 8\pi\, C_K\, h^{d-1} \,\theta_\alpha \Nrm{\op}{\L^p} \Nrm{\op_{\alpha:1}}{\L^{p'}}\, .
		\end{equation}
		Applying the Schatten space embedding inequality
		\begin{equation*}
			\Nrm{\op_{\alpha:1}}{\L^p} \leq h^{-d/p} \Nrm{\op_{\alpha:1}}{\L^1} = h^{-d/p},
		\end{equation*}
		to the first inequality in Formula~\eqref{est:derivative_schatten_p_op} and to Inequality~\eqref{est:derivative_theta_alpha} yields the following system of differential inequalities
		\begin{equation*}
			\begin{cases}
				u' \leq A\,v, 
				\\
				v' \leq a\,u\, v
			\end{cases}
		\end{equation*}
		where $u(t)=\Nrm{\op(t)}{\L^p}$ and $v(t)=\theta_\alpha(t)$ with $A = \frac{4\pi\,C_K}{Nh^{1+d/p}}$ large and $a = 8\pi\,C_K\, h^{d/p-1}$ small. Setting $U(t) = u(t/a)$ and $V(t) = \frac{A}{a}\,v(t/a)$, it can be written
		\begin{equation*}
			\begin{cases}
				U' \leq V, 
				\\
				V' \leq  U\,V.
			\end{cases}
		\end{equation*}
		It implies, for instance, that 
		\begin{equation*}
			\(U^2 + V\)' = 3\,U\,V \leq U^3 + 2\,V^{3/2} \leq 2\(U^2 + V\)^{3/2}.
		\end{equation*}
		Hence $\(U(t)^2 + V(t)\)^{-1/2} \geq \(U(0)^2 + V(0)\)^{-1/2} - t$, that is
		\begin{equation*}
			\Nrm{\op}{\L^p}^2 + A\,\theta_\alpha/a \leq \frac{\Nrm{\op^\init}{\L^p}^2 + A\,\theta_\alpha^\init/a}{\(1 - \(\Nrm{\op^\init}{\L^p}^2 + A\,\theta_\alpha^\init/a\)^{1/2} a\,t\)^2}
		\end{equation*}
		and Formula~\eqref{eq:schatten_propag_2} follows from the fact that $\frac{a}{A} = 2\,N\,h^{2d/p}$.
	\end{proof}
	
 \section{Proof of the main results}\label{sec:proof}
 
 	In this section we prove Theorem~\ref{thm:main} and Theorem~\ref{thm:main_2} and so we shall estimate the semiclassical optimal transport pseudo-metric between solutions of the BdG equation \eqref{eq:semiclassical_BdG} and the corresponding proposed classical coupled equations \eqref{eq:easy_equations} or \eqref{eq:Vlasov_F_error_1}.
 
 \subsection{A dynamics for the couplings}
 
	A coupling associated to $F$ and $\op_\alpha$ is a measurable function $\opUp:(z_1, z_2)\mapsto \opUp(z_1,z_2)$ defined for almost all $(z_1, z_2) \in \Rdd\times\Rdd$ with values in the space of bounded linear operators acting on $\h\otimes\h$ such that for almost all $(z_1, z_2) \in \Rdd\times\Rdd$, $\opUp(z_1, z_2) \ge 0$ and
	\begin{equation*}
		h^{2d}\tr_{\h\otimes\h}\!\(\opUp(z_1, z_2)\) = F(z_1, z_2)\quad\text{ and }\quad \intggam \opUp(z_1, z_2)\d z_1\d z_2 = \op_\alpha\,.
	\end{equation*}
	Then the semiclassical optimal transport pseudo-metric between $F$ and $\op_\alpha$ is
	\begin{equation}\label{eq:wasserstein-2}
	\Wh(F, \op_\alpha) :=\(\inf_{\opgam \in \cC(F, \op_\alpha)}\intggam h^{2d}\tr_{\h\otimes\h}\!\(\opC(z_1, z_2)\,\opUp(z_1, z_2)\)\d z_1\d z_2\)^\frac12
	\end{equation} 
	where $\opC(z_1, z_2) :=\opc(z_1)\otimes\id+\id\otimes\opc(z_2)$ that is 
	\begin{multline*}
		\opC(z_1, z_2)\Psi(x_1, x_2)
		= \(\n{\chi_1-x_1}^2+\n{\xi_1-\opp_1}^2\)\Psi(x_1, x_2)\\
		+\(\n{\chi_2-x_2}^2+\n{\xi_2-\opp_2}^2\)\Psi(x_1, x_2)\, .
	\end{multline*}

	For all $\opgam^\init \in \cC(f^\init,\op^\init )$ and $\opUp^\init \in \cC(F^\init, \op_\alpha^\init)$, let $(\opgam, \opUp)$ be the solution to the Cauchy problem
	\begin{subequations}\label{eq:cauchy_problem_for_coupling}
	\begin{equation}
		\dpt \opgam = \{H_f,\opgam\} +\frac{1}{i\hbar}\com{\sfH_{\op},\opgam} + \frac{1}{i\hbar} \frac{\theta_\alpha}{N}\intgam h^d\tr_2\!\(\com{K_{12},\opUp(z_1, z_2)}\)\d z_2
	\end{equation}
	and
	\begin{multline}
		\dpt \opUp = \{H_{f_{12}}, \opUp\} + \tfrac{\eta}{N} \nabla K(\chi_1-\chi_2)\cdot \(\nabla_{\xi_1}-\nabla_{\xi_2}\)\opUp\\
		+ \frac{1}{i\hbar}\com{\sfH_{\op_{12}}+K_{12}\(\tfrac{1}{N}-h^d\op_{12}\), \opUp}
		+ \frac{h^d}{i\hbar} \(\com{K_{12},\op_{12}} -\bangle{\com{K_{12},\op_{12}}}_{\op_{\alpha}}\)\opUp
	\end{multline}
	\end{subequations}
	with $(\opgam(0), \opUp(0)) = (\opgam^\init, \opUp^\init)$ and $\eta \in \set{0,1}$. More precisely, we will set $\eta = 0$ to prove Theorem~\ref{thm:main} and $\eta=1$ to prove Theorem~\ref{thm:main_2}. Notice that, in complete analogy with the well-posedness theory for the system \eqref{eq:Vlasov_and_error_1}--\eqref{eq:F_and_error_1}, one deduces the existence of the coupling dynamics $(\opgam, \opUp)$. It is then not difficult to see that with the above equations, the property of being a coupling is kept along the dynamics.
 
 \subsection{Estimating the semiclassical optimal transport pseudo-metrics}
 
	Let us now define the quantities
	\begin{align}
		\cE_{\opgam}(t) &:= \intgam h^d\tr_1\!\(\opc(z_1)\,\opgam(z_1)\)\d z_1 \label{def:E_gam}
		\\
		\cE_{\opUp}(t) &:= \intggam h^{2d}\tr_{12}\!\(\opC(z_1, z_2)\,\opUp(z_1, z_2)\)\d z_1\d z_2 \label{def:E_up}\, ,
	\end{align}
	where $\tr_1 = \tr_\h$ and $\tr_{12} = \tr_{\h\otimes\h}$. Since $(\opgam, \opUp)$ is a solution to the coupling Cauchy problem~\eqref{eq:cauchy_problem_for_coupling}, one obtains the following equations
	\begin{subequations}\label{eq:dE_gam_identity}
		\begin{align}
			\frac{\d\cE_{\opgam}(t)}{\d t} &= \intgam h^d\tr_1\!\(\{\opc(z_1), H_f\}\,\opgam(z_1)\)\d z_1\\
			&\quad + \frac{1}{i\hbar}\intgam h^d\tr_1\!\( \com{\opc(z_1), \tfrac{1}{2}\n{\opp}^2+ V_{\op}}\opgam(z_1)\)\d z_1\\
			&\quad - \frac{h^d}{i\hbar}\intgam h^d\tr_1\!\( \com{\opc(z_1), \sfX_{\op}}\opgam(z_1)\)\d z_1 \label{def:dE_Fock_term}\\
			&\quad + \frac{1}{i\hbar} \frac{\theta_\alpha}{N}\intggam h^{2d}\tr_{12}\!\(\com{\opc(z_1)\otimes \id, K_{12}}\opUp(z_1, z_2)\)\d z_1\d z_2 \label{def:dE_term3}
		\end{align}
	\end{subequations}
	and, using the fact that 
	\begin{align*}
		\opC(z_1, z_2)\com{K_{12},\op_{12}}-\com{\opC(z_1, z_2),K_{12}\op_{12}}
		 = \opC(z_1, z_2)\,\op_{12}\,K_{12} - K_{12}\,\op_{12}\,\opC(z_1, z_2)
	\end{align*}
	we can write
	\begin{subequations}\small\label{eq:dE_Up_identity}
		\begin{align}\label{def:dE_up_term1}
			\frac{\d\cE_{\opUp}(t)}{\d t} &= \intggam h^{2d}\tr_{12}\!\(\{\opC(z_1, z_2), H_{f_{12}}\}\, \opUp(z_1, z_2)\)\d z_1\d z_2
			\\\label{def:dE_up_term2}
			&\quad + \frac{1}{i\hbar}\intggam h^{2d}\tr_{12}\!\( \com{\opC(z_1, z_2), \sfH_{\op_{12}}}\opUp(z_1, z_2)\)\d z_1\d z_2
			\\\label{def:dE_up_term3}
			&\quad - \frac{\eta}{N}\intggam h^{2d}\tr_{12}\!\(\nabla_{\xi_{12}}\opC(z_1, z_2)\cdot\nabla K_{\chi_{12}} \opUp(z_1, z_2)\)\d z_1\d z_2
			\\\label{def:dE_up_term4}
			&\quad + \frac{1}{i\hbar\, N}\intggam h^{2d}\tr_{12}\!\( \com{\opC(z_1, z_2),K_{12}}\opUp(z_1, z_2)\)\d z_1\d z_2
			\\\label{def:dE_up_term5}
			&\quad - \frac{2\,h^d}{\hbar} \im\intggam h^{2d}\tr_{12}\!\( K_{12}\,\op_{12}\,\opC(z_1, z_2)\,\opUp(z_1, z_2)\)\d z_1\d z_2
			\\\label{def:dE_up_term6}
			&\quad - \frac{h^d}{i\hbar} \bangle{\com{K_{12},\op_{12}}}_{\op_{\alpha}} \intggam h^{2d}\tr_{12}\!\(\opC(z_1, z_2) \opUp(z_1, z_2)\)\d z_1\d z_2\, .
		\end{align}
	\end{subequations}
	
	\subsubsection{Estimates for \texorpdfstring{$\cE_{\opgam}$}{Egam}} To estimate the right-hand side of Identity~\eqref{eq:dE_gam_identity}, let us focus on Term~\eqref{def:dE_Fock_term} since the first two terms are already handled in \cite[Theorem 2.5]{golse_schrodinger_2017}. For Term~\eqref{def:dE_Fock_term}, we notice that 
	\begin{equation*}
		\com{\n{\chi-x}^2, \sfX_{\op}}(x, y) = \tfrac12 \,K(x-y)\((x-y)-2(\chi-y)\)\cdot\(x-y\)\op(x, y)
	\end{equation*}
	which yields the estimate 
	\begin{multline}\label{est:commutator_spatial_exchange-term_Egam}
		\n{\frac{h^d}{i\hbar}\intgam h^d\tr_1\!\( \com{\n{\chi_1-x}^2, \sfX_{\op}}\opgam(z_1)\)\d z_1}
		\\
		\le Ch^{d-1}\Nrm{\n{\cdot} K(\cdot)}{L^\infty}\Nrm{[x, \op]}{\L^\infty}+Ch^{d-1}\Nrm{K}{L^\infty}\Nrm{[x,\op]}{\L^\infty}\cE_{\opgam}\, .
	\end{multline}
	Next, to estimate the part involving $\com{\n{\xi-\opp}^2, \sfX_{\op}}$, recall the identity
	\begin{align*}
		[\n{\opp-\xi}^2, \sfX_{\op}] &= (\opp-\xi)\cdot[\opp, \sfX_{\op}]+[\opp, \sfX_{\op}]\cdot (\opp-\xi)\\
		&= \(\opp-\xi\)\cdot \sfX_{\com{\opp, \op}}+\sfX_{\com{\opp, \op}}\cdot \(\opp-\xi\) ,
	\end{align*}
	and the fact that for any $\op \in \opP(\h)$ and $A, B$ are self-adjoint possibly unbounded operators, we have that 
	\begin{equation*}
		\Tr{\(AB+BA\)\op}\le \Tr{\(A^2+B^2\)\op}.
	\end{equation*}
	Then it follows that 
	\begin{equation}\label{est:commutator_moments_exchange-term_Egam}
		\n{\frac{h^d}{i\hbar}\intgam h^d\tr_1\!\( [\n{\xi_1-\opp}^2, \sfX_{\op}]\opgam(z_1)\)\d z_1} \le\, Ch^{2d-2}\Nrm{ \sfX_{[\opp, \op]}}{\L^\infty}^2+C\,\cE_{\opgam}\, .
	\end{equation}
	By Inequality \eqref{eq:Xop}, H\"older's inequality, and the semiclassical Schatten space embedding inequality, we have that 
	\begin{equation}\label{est:commutator_moments_exchange-term_Egam_2}
		\begin{aligned}
			\mathrm{RHS}\eqref{est:commutator_moments_exchange-term_Egam} &\le C\,h^{2d-2}\,\snrm{\hat K}_{L^1}^2\Nrm{[\opp, \op]}{\L^\infty}^2+C\,\cE_{\opgam}
			\\
			&\le 2C'_K h^{2d-2}\Nrm{\opp\sqrt{\op}}{\L^\infty}^2\Nrm{\op}{\L^\infty}+C\,\cE_{\opgam}
			\\
			&\le 2C'_K h^{3d/2-3}\Nrm{\opp\sqrt{\op}}{\L^4}^2\Nrm{\op}{\L^d}+C\,\cE_{\opgam}
			\\
			&\le 2C'_K h^{(3d-7)/2}\(h^{d}\Tr{\n{\opp}^2\op \n{\opp}^2}\)^\frac12\Nrm{\op}{\L^d}^{3/2}+C\,\cE_{\opgam}\, .
		\end{aligned}
	\end{equation}
	We proceed similarly to estimate $\Nrm{\com{x,\op}}{\L^\infty}$ in Inequality~\eqref{est:commutator_spatial_exchange-term_Egam}, that is we write
	\begin{equation}\label{eq:com_x_op}
		h^{d-1} \Nrm{\com{x,\op}}{\L^\infty} \leq Ch^{(3d-7)/2}\(h^{d}\Tr{\n{x}^2\op \n{x}^2}\)^{1/2} \Nrm{\op}{\L^d}^{3/2}.
	\end{equation}
	Now, combining Inequalities \eqref{est:commutator_spatial_exchange-term_Egam} and \eqref{est:commutator_moments_exchange-term_Egam_2}, we obtain the following bound
	\begin{multline*}
		\n{\eqref{def:dE_Fock_term}} \le C_K\, h^\frac{3d-7}{2}\(M_4+N_4\)^{1/2}
		\Nrm{\op}{\L^d}^{3/2}
		\\
		\qquad + C_K \(1+h^\frac{3d-7}{2} \(M_4+N_4\)^{1/2} \Nrm{\op}{\L^d}^{3/2}\)\cE_{\opgam}
	\end{multline*}
	where $M_4 = \Tr{\op\n{\opp}^4}$ and $N_4 = \Tr{\op\n{x}^4}$.
	
	Next, notice that 
	\begin{align}\label{eq:cost_function-interaction_commutator_identity}
		\tfrac{1}{i\hbar}\com{\opc(z_1)\otimes\id, K_{12}}= (\xi_1-\opp_1)\cdot \grad K_{12}+ \grad K_{12} \cdot (\xi_1-\opp_1)\, ,
	\end{align}
	where $x_2$ is viewed as a constant, then we can rewrite Term~\eqref{def:dE_term3} as follows
	\begin{subequations}
		\begin{align}\notag
			\n{\eqref{def:dE_term3}} &= \frac{2\,\theta_\alpha}{N}\re \intggam \!\!h^{2d}\tr_{12}\!\((\xi_1-\opp_1)\cdot \grad K_{12}\,\opUp(z_1, z_2)\)\d z_1\d z_2
			\\\label{def:dE_new_term1}
			&\le \frac{\theta_\alpha}{N}\intggam h^{2d}\tr_{12}\!\(\n{\xi_1-\opp_1}^2\opUp(z_1, z_2)\)\d z_1\d z_2
			\\\label{def:dE_new_term2}
			&\quad +\frac{\theta_\alpha}{N}\intggam h^{2d}\tr_{12}\!\(\n{\grad K_{12}}^2\opUp(z_1, z_2)\)\d z_1\d z_2.
		\end{align}
	\end{subequations}
	It is clear that~Term \eqref{def:dE_new_term1} is bounded by $\frac{\theta_\alpha}{N} \, \cE_{\opUp}$. For Term~\eqref{def:dE_new_term2}, we use the fact that $\grad K_{12}$ is a bounded multiplication operator of norm $\Nrm{\grad K_{12}}{L^\infty}$, H\"older's inequality, and the fact that $h^{2d}\tr_{12}\!\(\opUp(z_1, z_2)\) = F(z_1,z_2)$ has integral one on $\Rdd\times\Rdd$ to deduce that Term~\eqref{def:dE_new_term2} is bounded by $\Nrm{\grad K}{L^\infty}^2\frac{\theta_\alpha}{N}$. Hence combining our calculations with the result in \cite{golse_schrodinger_2017}, we obtain the bound 
	\begin{multline}\label{est:dE_gam-inequality}
		\frac{\d\cE_{\opgam}}{\d t} \le \(C_K'+C_K h^{(3d-7)/2}(M_4+N_4)^{1/2} \Nrm{\op}{\L^d}^{3/2}\) \cE_{\opgam} + \tfrac{\theta_\alpha}{N}\,\cE_{\opUp}
		\\
		+\tfrac{\theta_\alpha}{N}\Nrm{\grad K}{L^\infty}^2+C_K h^{(3d-7)/2}\(M_4+N_4\)^{1/2}
		 \Nrm{\op}{\L^d}^{3/2},
	\end{multline}
	where $C_K'$ depends on the uniform bound of $\grad^2K$. 
	
	\subsubsection{Estimates for \texorpdfstring{$\cE_{\opUp}$}{Eup}} To estimate the right hand side of Equation \eqref{eq:dE_Up_identity}, we need the following identities
	\begin{align}
		\{\opC(z_1, z_2), H_{f_{12}}\} &= \{\opc(z_1), H_f(z_1)\}\otimes \id +\id\otimes\{\opc(z_2), H_f(z_2)\}
		\\
		\com{\opC(z_1, z_2), \sfH_{\op_{12}}}&= \com{\opc(z_1), \sfH_{\op}}\otimes \id+\id\otimes \com{\opc(z_2), \sfH_{\op}} ,
	\end{align}
	and
	\begin{multline}\label{eq:two-variable_cost_function_commutator_identity}
		\com{\opC(z_1, z_2),K_{12}\op_{12}}= \com{\opc(z_1)\otimes\id, K_{12}}\op_{12}+\com{\id\otimes\opc(z_2), K_{12}}\op_{12}\\
		+K_{12}\(\com{\opc(z_1), \op}\otimes \id+\id\otimes \com{\opc(z_2), \op}\).
	\end{multline}
	For the first term, we see that 
	\begin{align*}
		\n{\eqref{def:dE_up_term1}} &=
		\intggam h^{2d} \tr_{12}\!\(\(\{\opc(z_1), H_f(z_1)\}\otimes\id\) \opUp(z_1, z_2)\) \d z_1\d z_2
		\\
		&\quad + \intggam h^{2d} \tr_{12}\!\(\(\id\otimes\{\opc(z_2), H_f(z_2)\}\) \opUp(z_1, z_2)\) \d z_1\d z_2
		\\
		&= 2\intgam h^d\tr_1\!\(\{\opc(z), H_f(z)\} \opUp_{:1}(z)\)\d z 
	\end{align*}
	where 
	\begin{align*}
		\opUp_{:1}(z_1):= \intgam h^d\tr_2\!\(\opUp(z_1, z_2)\)\d z_2\, . 
	\end{align*}
	Then using the same argument as in \cite{golse_schrodinger_2017}, we obtain the bound 
	\begin{align*}
		\n{\eqref{def:dE_up_term1}} \le&\, 2\(1+\max\(4\Nrm{\grad^2 K}{L^\infty}^2, 1\)\)\intdd h^d\tr_1\!\(\opc(z) \opUp_{:1}(z)\)\d z \\
		&=  2\(1+\max\(4\Nrm{\grad^2 K}{L^\infty}^2, 1\)\)\cE_{\opUp}(t)\, .
	\end{align*}
	A similar argument holds for Term~\eqref{def:dE_up_term2} with minor modification for the term $h^d\sfX_{\op}$ as seen in above estimate for Term~\eqref{def:dE_Fock_term}.
	
	Terms~\eqref{def:dE_up_term3} and \eqref{def:dE_up_term4} follows the same argument as in the case of Term~\eqref{def:dE_term3}, that is, 
	\begin{align*}
		\n{\eqref{def:dE_up_term3}+\eqref{def:dE_up_term4}}
		&\le \frac{1}{N}\intggam h^{2d}\tr_{12}\!\(\n{\xi_1-\opp_1}^2\opUp(z_1, z_2)\)\d z_1\d z_2\\
		&\quad +\frac{1}{N}\intggam \!h^{2d}\tr_{12}\!\(\n{\xi_2-\opp_2}^2\opUp(z_1, z_2)\)\d z_1\d z_2\\
		&\quad +\frac{1}{N}\intggam \!h^{2d}\tr_{12}\!\(\n{\grad K_{12}-\eta\,\grad K_{\chi_{12}}}^2\opUp(z_1, z_2)\)\d z_1\d z_2\, .
	\end{align*}
	If $\eta = 1$, then we use the fact that
	\begin{align*}
		\n{\grad K(x_1-x_2)-\grad K(\chi_1-\chi_2)}^2 \le 2 \Nrm{\grad^2 K}{L^\infty}^2\(\n{x_1-\chi_1}^2+\n{x_2-\chi_2}^2\),
	\end{align*}
	to obtain
	\begin{equation*}
		\n{\eqref{def:dE_up_term3}+\eqref{def:dE_up_term4}} \le \frac{2\max\(\Nrm{\grad^2 K}{L^\infty}^2, 1\)}{N}\,\cE_{\opUp}(t)\, .
	\end{equation*}
	If $\eta = 0$, then we use instead the fact that $\nabla K_{12}$ is a bounded multiplication operator to get
	\begin{equation*}
		\n{\eqref{def:dE_up_term3}+\eqref{def:dE_up_term4}} \le \frac{1}{N}\(\cE_{\opUp}(t) + \Nrm{\nabla K}{L^\infty}\).
	\end{equation*}
	
	For Term \eqref{def:dE_up_term6}, we notice that 
	\begin{equation*}
		\n{\bangle{\com{K_{12},\op_{12}}}_{\op_{\alpha}}} = \n{h^{2d}\tr_{12}\(\com{K_{12},\op_{12}}\op_\alpha\)} \le C\Nrm{K}{L^\infty}\Nrm{\op}{\L^\infty},
	\end{equation*}
	then this yields the bound
	\begin{equation*}
		\n{\eqref{def:dE_up_term6}}\le Ch^{d-1}\Nrm{K}{L^\infty}\Nrm{\op}{\L^\infty}\cE_{\opUp}(t)\,.
	\end{equation*}
	
	Finally, to handle Term \eqref{def:dE_up_term5}, we start by expanding the expression 
	\begin{equation*}
		K_{12}\,\op_{12}\,\opC(z_1, z_2)=K_{12}\,\(\op_1\,\opc(z_1)+\opc(z_1)\op_2+\op_1\opc(z_2)+\op_2\,\opc(z_2)\).
	\end{equation*}
	where $\op_1 = \op\otimes \id$ and $\op_2 = \id\otimes\op$. 
	It suffice to consider the first two terms in the above expansion since the others are handled in the exact same manner, i.e. we estimate 
	\begin{subequations}
		\begin{align}\label{def:dE_up_term5:1}
			&\frac{h^d}{\hbar} \im\intggam h^{2d}\tr_{12}\!\( K_{12}\,\op_1\,\opc(z_1)\,\opUp(z_1, z_2)\)\d z_1\d z_2\, ,
			\\\label{def:dE_up_term5:2}
			&\frac{h^d}{\hbar} \im\intggam h^{2d}\tr_{12}\!\( K_{12}\,\opc(z_1)\,\op_2\,\opUp(z_1, z_2)\)\d z_1\d z_2\, .
		\end{align}
	\end{subequations}
	In the first case, notice that 
	\begin{subequations}
		\begin{multline}\label{def:dE_up_term5:1a}
			\frac{h^d}{\hbar} \im\intggam h^{2d}\tr_{12}\!\( K_{12}\,\op_1\,\n{\chi_1-x_1}^2\,\opUp(z_1, z_2)\)\d z_1\d z_2
			\\
			=\frac{h^d}{\hbar} \im\intggam h^{2d}\tr_{12}\!\( K_{12}\com{x_1, \op_1} \cdot(\chi_1-x_1)\,\opUp(z_1, z_2)\)\d z_1\d z_2
			\\
			\quad +\frac{h^d}{\hbar} \im\intggam h^{2d}\tr_{12}\!\(K_{12}\,\op_1(\chi_1-x_1)\,\opUp(z_1, z_2)\cdot  (\chi_1-x_1)\)\d z_1\d z_2
		\end{multline}
	from which it follows
	\begin{align*}
		\n{\eqref{def:dE_up_term5:1a}}\le&\, Ch^{d-1}\Nrm{K}{L^\infty}\Nrm{[x, \op]}{\L^\infty}\cE_{\opUp}^{1/2}+Ch^{d-1}\Nrm{K}{L^\infty}\Nrm{\op}{\L^\infty}\cE_{\opUp}\\
		\le&\, C_K\, h^{(3d-7)/2}\, N_4^{1/2}\Nrm{\op}{\L^d}^{3/2}+C\(1+h^{d-1}\Nrm{K}{L^\infty}\Nrm{\op}{\L^\infty}\)\cE_{\opUp},
	\end{align*}
	where the second inequality follows the same argument as in~\eqref{est:commutator_moments_exchange-term_Egam_2}.
	Similarly, we see that
	\begin{multline}\label{def:dE_up_term5:1b}
		\frac{h^d}{\hbar} \im\intggam h^{2d}\tr_{12}\!\( K_{12}\,\op_1\,\n{\xi_1-\opp_1}^2\,\opUp(z_1, z_2)\)\d z_1\d z_2
		\\
		=\frac{h^d}{\hbar} \im\intggam h^{2d}\tr_{12}\!\( K_{12}[\opp_1, \op_1]\cdot(\xi_1-\opp_1)\,\opUp(z_1, z_2)\)\d z_1\d z_2
		\\
		\quad - h^d \re\intggam h^{2d}\tr_{12}\!\( \grad K_{12}\op_1\cdot(\xi_1-\opp_1)\,\opUp(z_1, z_2)\)\d z_1\d z_2
		\\
		\quad +\frac{h^d}{\hbar} \im\intggam h^{2d}\tr_{12}\!\(K_{12}\,\op_1(\xi_1-\opp_1)\,\opUp(z_1, z_2)\cdot  (\xi_1-\opp_1)\)\d z_1\d z_2
	\end{multline}
	\end{subequations}
	then, by a same argument as Inequality~\eqref{est:commutator_moments_exchange-term_Egam_2}, we have that
	\begin{multline*}
		\n{\eqref{def:dE_up_term5:1b}}\le 
		Ch^{(3d-7)/4}\Nrm{K}{L^\infty}\Nrm{\op}{\L^d}^{3/4}\,M_4^{1/4}\cE_{\opUp}^{1/2}\\
		 +Ch^{d}\Nrm{\grad K}{L^\infty}\Nrm{\op}{\L^\infty}\cE_{\opUp}^{1/2}+Ch^{d-1}\Nrm{K}{L^\infty}\Nrm{\op}{\L^\infty}\cE_{\opUp}\,.
	\end{multline*}
	This completes the estimate for Term \eqref{def:dE_up_term5:1}. 
	
	To estimate Term \eqref{def:dE_up_term5:2}, we follow a similar idea as above. We write 
	\begin{subequations}
	\begin{multline}\label{def:dE_up_term5:2a}
		\frac{h^d}{\hbar} \im\intggam h^{2d}\tr_{12}\!\( K_{12}\,\n{\chi_1-x_1}^2\,\op_2\,\opUp(z_1, z_2)\)\d z_1\d z_2
		\\
		=\frac{h^d}{\hbar}\im\intggam h^{2d}\tr_{12}\!\( K_{12} \,\op_2 (\chi_1-x_1)\,\opUp(z_1, z_2)\cdot (\chi_1-x_1) \)\d z_1\d z_2
	\end{multline}
	which leads to
	\begin{align*}
		\n{\eqref{def:dE_up_term5:2a}}\le Ch^{d-1}\Nrm{K}{L^\infty}\Nrm{\op}{\L^\infty}\cE_{\opUp}\,.
	\end{align*}
	Next, we have 
	\begin{multline}\label{def:dE_up_term5:2b}
		\frac{h^d}{\hbar} \im\intggam h^{2d}\tr_{12}\!\( K_{12}\,\n{\xi_1-\opp_1}^2\,\op_2\,\opUp(z_1, z_2)\)\d z_1\d z_2
		\\
		=\frac{h^d}{\hbar}\im\intggam h^{2d}\tr_{12}\!\( K_{12} \,\op_2 (\xi_1-\opp_1)\,\opUp(z_1, z_2)\cdot (\xi_1-\opp_1) \)\d z_1\d z_2
		\\
		\quad -h^d\re\intggam h^{2d}\tr_{12}\!\( \grad K_{12}\,\op_2\cdot  (\xi_1-\opp_1)\,\opUp(z_1, z_2) \)\d z_1\d z_2
	\end{multline}
	\end{subequations}
	which yields 
	\begin{equation*}
		\n{\eqref{def:dE_up_term5:2b}}\le Ch^{d-1}\Nrm{K}{L^\infty}\Nrm{\op}{\L^\infty}\cE_{\opUp}+Ch^{d}\Nrm{\grad K}{L^\infty}\Nrm{\op}{\L^\infty}\cE_{\opUp}^{1/2}.
	\end{equation*}
	Hence we obtain the following bound 
	\begin{multline*}
		\n{\eqref{def:dE_up_term5}} \le C \(1+h^{d-1}\Nrm{K}{L^\infty}\Nrm{\op}{\L^\infty}+h^{d}\Nrm{\grad K}{L^\infty}\Nrm{\op}{\L^\infty}\) \cE_{\opUp}
		\\
		+ C_K \(h^{d}\Nrm{\op}{\L^\infty} + h^{(3d-7)/2}\(N_4+M_4\)^{1/2}\Nrm{\op}{\L^d}^{3/2}\) .
	\end{multline*}

	Finally, combining the above estimates, we see that there exists a constant $C$, dependent on $K$, such that we have the following inequality
	\begin{multline}\label{est:dE_up-inequality}
		\dt\cE_{\opUp}(t) \le C_{K}'\(1+h^{(3d-7)/2}(M_4+N_4)^{1/2} \Nrm{\op}{\L^d}^{3/2}+h^{d-1}\Nrm{\op}{\L^\infty}\) \cE_{\opUp}(t)
		\\
		+ C_K \(\frac{1-\eta}{N}+h^{d}\Nrm{\op}{\L^\infty} + h^{(3d-7)/2}\(N_4+M_4\)^{1/2}\Nrm{\op}{\L^d}^{3/2}\) .
	\end{multline}

	In the case when $Nh$ is bounded from below by a constant independent of $N$ and $\hbar$, then it follows from the last inequality in Formula~\eqref{eq:ineq_alpha_2} that
	\begin{equation}
		h^{d-1} \Nrm{\op}{\L^\infty} \leq \frac{1}{Nh}
	\end{equation}
	is bounded uniformly in $\hbar$ and $N$. Moreover, by Proposition~\ref{prop:semiclassical_propagation_of_moments} and Proposition~\ref{prop:schatten_propag}~\ref{part:N2h>1_regime}, we see that $\Nrm{\op}{\L^d}, M_4,$ and $N_4$ are propagated uniformly in $N$ and $\hbar$. Then, by Inequalities~\eqref{est:dE_gam-inequality} and \eqref{est:dE_up-inequality}, we see there exists a constant $C_{K, \op}$, depend on $K$ and $\op$,  such that 
	\begin{align*}
		\dt\(\cE_{\opgam}(t)+\cE_{\opUp}(t)\) &\le C_{K, \op} \(1+h^{(3d-7)/2}(M_4+N_4)^{1/2}\)\(\cE_{\opgam}(t)+\cE_{\opUp}(t)\)
		\\
		&\qquad+ C_{K,\op}\,h \(\tfrac{\theta_\alpha+1}{N\, h}+h^{(3d-9)/2}\(N_4+M_4\)^{1/2}\)
		\\
		&\le C_{K, \op}\,g_1(t) \(\cE_{\opgam}(t)+\cE_{\opUp}(t)\)+C_{K, \op}\,h\,g_0(t)\,, 
	\end{align*}
	where $g_i(t)\ge 1+ h^{i+(3d-9)/2}\(N_4+M_4\)^{1/2}$. Recalling the definitions~\eqref{eq:wasserstein-1} and~\eqref{eq:wasserstein-2}, we conclude the proof of Theorem~\ref{thm:main} by Gr\"{o}nwall's lemma.

	In the case when $Nh\ll 1$, $\theta_\alpha^\init \leq C \,N\,h^{2d/p}$ and $\Nrm{\op^\init}{\L^p} \leq C$ for some constant $C$ independent of $\hbar$, then it follows from Proposition~\ref{prop:schatten_propag} that there exists $C>0$ independent of $\hbar$ such that for any $t\in[0,T]$ with $T = C\,h^{1-d/p}$,
	\begin{equation*}
		h^{d-1} \Nrm{\op}{\L^\infty} \leq h^{\frac{d}{p'}-1}\Nrm{\op}{\L^p} \leq C\, .
	\end{equation*}
	The moments also remain propagated uniformly in $\hbar$ and $N$ in this case. Then, by Inequalities \eqref{est:dE_gam-inequality} and \eqref{est:dE_up-inequality}, we see there exists $C_{K,\op}$, depending on $K$ and $\Nrm{\op}{\L^d}$,  such that 
	\begin{align*}
		\dt\(\cE_{\opgam}(t)+\cE_{\opUp}(t)\) &\le C_{K, \op} \(1+h^{(3d-7)/2}(M_4+N_4)^{1/2}\)\(\cE_{\opgam}(t)+\cE_{\opUp}(t)\)\\
		&\qquad + C_{K,\op}\,h \(\tfrac{\theta_\alpha}{N\,h}+h^{(3d-9)/2}\(N_4+M_4\)^{1/2}\)
		\\
		&\le C_{K,\op}'\,g_1(t) \(\cE_{\opgam}(t)+\cE_{\opUp}(t)\)+C_{K, \op}'\,h\,g_0(t)\,, 
	\end{align*}
	where $g_i(t)\ge 1+ h^{i+(3d-9)/2}\(N_4+M_4\)^{1/2}$, for $i\in\{0,1\}$. Notice that the last inequality is possible since $\theta_\alpha/(N\,h)\le C h$ on  $[0, T]$. Again, we conclude the proof of Theorem~\ref{thm:main_2} by Gr\"{o}nwall's lemma.

\section{Application to the effective approximation of quantum systems}\label{sec:application}

	In this section, we combine the result from the previous section and the result in \cite{marcantoni_dynamics_2023}. To avoid a substantial detour from the goal of the paper, we will provide a concise introduction on the method of second quantization, covering only the essential definitions necessary for stating the main result of \cite[Theorem 3.3]{marcantoni_dynamics_2023}. Also, in this section, we assume the scaling $N\, h^d =1$ with $d=3$.

\subsection{Quasi-free approximation of interacting spin-\texorpdfstring{$\frac12$}{1/2} fermions}

	Let $\hh$ denote the complex Hilbert space $L^2(X)$ where $X=\Rd \times\{\uparrow, \downarrow\}$. The elements of $X$ are expressed as ordered pairs $\x = (x, \tau)$ where $x \in \Rd$ is the spatial variable and $\tau \in \{\uparrow, \downarrow\}$ is called the spin label. Notice we have the identification $\hh \cong L^2(\Rd)\otimes \CC^2$. Let $\hh^{\wedge n} := \hh\wedge \cdots \wedge \hh$ denote the $n$-fold anti-symmetric tensor product. We define the fermionic Fock space $\cF$ over $\hh$ to be the closure of algebraic direct sum
	\begin{equation}
		\cF^{\mathrm{alg}}(\hh) := \CC\oplus \bigoplus_{n=1}^\infty \hh^{\wedge n}
	\end{equation}
	with respect to the norm $\Nrm{\cdot}{\cF}$ induced by the endowed inner product
	\begin{equation}
		\Inprod{\Psi}{\Phi} = \conj{\psi^{(0)}}\, \varphi^{(0)} + \sum_{n\geq 1} \Inprod{\psi^{(n)}}{\varphi^{(n)}}_{\hh^{\otimes n}},
	\end{equation}
	for any pair of vectors $\Psi = (\psi^{(0)}, \psi^{(1)}, \ldots)$ and $\Phi = (\varphi^{(0)}, \varphi^{(1)}, \ldots)$ in $\cF^{\mathrm{alg}}(\hh)$. A normalized vector $\Psi$ in $\cF$ is called a Fock state or a pure state. The vacuum, defined by the vector $\Omega_{\cF} = (1, 0, 0, \ldots)\in\cF,$
	describes the state with no particles.
	
	For every $\x \in X$, we define the corresponding creation and annihilation operator-valued distributions, denoted by $a^\ast_{\x}$ and $a_{\x}$, acting on $\cF$ by their actions on the $n$-sector of $\cF$ as follows
	\begin{align*}
		(a^*_{\x}\,\Psi)^{(n)}(\underline{\x}_n) &:= \frac{1}{\sqrt{n}}\sum^n_{j=1} \(-1\)^{j-1}\delta(x-x_j)\,\delta_{\tau, \tau_j}\, \psi^{(n-1)}(\underline{\x}_{n\backslash\, j}),
		\\
		(a_{\x}\,\Psi)^{(n)}(\underline{\x}_n) &:= \sqrt{n+1}\,\psi^{(n+1)}(\x, \underline{\x}_n),
	\end{align*}
	where $ \underline{\x}_{n}:= (\x_1, \ldots, \x_n), \underline{\x}_{n\backslash\, j}:= (\x_1, \ldots, \cancel{\x}_j, \ldots, \x_n),$ and $\delta_{\tau, \tau'}$ is the Kronecker delta. Moreover, the action of the annihilation operator on the vacuum of $\cF$ is defined to be $a_{\x}\,\Omega_{\cF} = 0$. Then, we extend the operators linearly to the whole $\cF$. It can easily be checked that the collection of creation and annihilation operators on $\cF$ satisfies the canonical anti-commutation relations
	\begin{equation}\label{eq:CAR}
		\com{a_{\x}, a^*_{\x'}}_+ = \delta(x-x')\delta_{\tau, \tau'}\, , \quad \com{a_{\x}, a_{\x'}}_+ = \com{a^*_{\x}, a^*_{\x'}}_+ = 0
	\end{equation}
	for all $\x, \x' \in X$ where $\com{A, B}_+ = AB + BA$ is the anti-commutator of the operators $A$ and $B$. Another useful operator is given by the number operator 
	\begin{equation*}
		\cN = \bigoplus^\infty_{n=1} n\,\id_{\hh^{\wedge n}}
	\end{equation*}
	which counts the number of particles in each sector.

	Consider the fermionic Fock state $\Psi\in \cF$ with an expected number of particles equals to $N$, i.e. 
	$\Inprod{\Psi}{\cN\,\Psi}_{\cF}= N$. We define its one-particle reduced density operator $\ome$ and its pairing operator $\al$ to be the operators with integral kernels 
	\begin{subequations}
		\begin{align}
			\ome(\x;\y) &:= \Inprod{\Psi}{a_{\y}^*\,a_{\x}\, \Psi},
			\\
			\al(\x, \y) &:= \Inprod{\Psi}{a_{\y}\,a_{\x}\, \Psi}, \quad \op_{\alpha,\Psi}(\x_{12};\y_{12}) := \frac{1}{h^d\,\theta_\Psi}\, \alpha(\x_1;\x_2)\,\alpha(\y_1;\y_2),
		\end{align}
	\end{subequations}
	where $\theta_\Psi = \frac{1}{N} \Nrm{\alpha_\Psi}{2}^2 = \Nrm{\alpha_\Psi}{\L^2}^2 \in [0,1]$ is such that $h^{2d}\Tr{\op_{\alpha,\Psi}} = 1$. Notice, we have that $h^d \tr_{\hh}\!\(\ome\) = 1$ while $\tr_{\hh}\(\al\) = 0$. Moreover, since we are in the case of fermions, it follows from Properties \eqref{eq:CAR} that $0\leq \ome\leq \id$ and $\al$ is anti-symmetric. More compactly, we introduce the generalized one-particle density operator acting on $\hh\oplus\hh$ by
	\begin{equation}\label{eq:Gamma_Psi}
		\Gamma_{\Psi} := \(\begin{array}{cc}
		\ome & \al \\
		\al^* & \id-\overline{\ome}
		\end{array}\), 
	\end{equation}
	which satisfies $0\leq \Gamma \leq \id_{\hh\oplus\hh}$.
	
	We say that $\Psi$ is a quasi-free pure state if
	\begin{equation*}
		\Inprod{\Psi}{a^{\sharp_1}_{\x_1}a^{\sharp_2}_{\x_2}\cdots a^{\sharp_{2\ell-1}}_{\x_{2\ell-1}}\,\Psi} = 0
	\end{equation*}
	and 
	\begin{equation*}
		\Inprod{\Psi}{a^{\sharp_1}_{\x_1}a^{\sharp_2}_{\x_2}\cdots a^{\sharp_{2\ell}}_{\x_{2\ell}}\,\Psi} = \sum_{\pi \in \mathbb{P}_{2\ell}} \operatorname{sgn}(\pi)\prod^{\ell}_{\jj=1}\Inprod{\Psi}{a^{\sharp_{\pi(2\jj-1)}}_{\x_{\pi(2\jj-1)}}a^{\sharp_{\pi(2\jj)}}_{\x_{\pi(2\jj)}}\,\Psi} 
	\end{equation*}
	for all $\ell \in \N$, where $a^\sharp$ denotes either $a^\ast$ or $a$ and $\mathbb{P}_{2\ell}$ is the set of pairings, that is, the subset of permutations of $2\ell$ elements satisfying 
	\begin{equation*}
		\pi(2\jj-1) < \pi(2\jj)\quad \forall\, \jj =1, \ldots, \ell \quad \text{ and } \quad 
		\pi(2\jj-1) < \pi(2\jj+1) \quad \forall\, \jj =1, \ldots, \ell-1.
	\end{equation*}
	In other words, observables associated with product of creation and annihilation operators of a quasi-free state are completely characterized by $\ome$ and $\al$. Moreover, if $\Psi$ is quasi-free, then $\Gamma_{\Psi}$ is a projection operator, or more precisely
	\begin{equation}\label{conditions:Gamma}
		0\le \op \le \id, \quad \alpha^\ast = -\conj{\alpha},\quad \op\,\alpha =\alpha\,\conj{\op}, \quad \text{ and } \quad \n{\alpha^\ast}^2 = \op(\id -\op).
	\end{equation} 
	Conversely, if $\Gamma$ is of the form \eqref{eq:expression_Gamma} satisfying Conditions \eqref{conditions:Gamma} and $\gamma$ is a trace class operator, then there exists a quasi-free pure state $\Psi$ such that $\Gamma_{\Psi}=\Gamma$ (cf. Chapter 10 or Appendix G in \cite{solovej_many_2014}). Furthermore, if $\Psi$ is a quasi-free pure state with finite expected number of particles, then there exists a unitary transformation $\sfR$, parameterized by $\ome$ and $\al$, such that $\Psi = \sfR\Omega_{\cF}$. $\sfR$ is a Bogoliubov transformation (cf. \cite{solovej_many_2014}).

	Let $K(x)$ be a spin-independent radial function. Define the Hamiltionian in the Fock space by
	\begin{equation*}
		\cH_N = \int_{X} a^\ast_{\x}\(-\tfrac{\hbar^2}{2}\lapl_x\) a_{\x}\, \mu(\d\x) +\frac{1}{2N}\int_{X\times X} K(x-y)\,a^\ast_{\x}a^\ast_{\y}a_{\y}a_{\x}\,\mu(\d\x)\,\mu(\d\y),
	\end{equation*}
	where $\mu$ is the tensor product of the Lebesgue measure on $\Rd$ and the counting measure, and consider the time-dependent Fock state $\Psi(t) = \Psi$ given by
	\begin{equation}\label{eq:many-body_state}
		\Psi = e^{-i \cH_N t/\hbar}\,\Psi^\init = e^{-i \cH_N t/\hbar}\,\sfR^\init \Omega_{\cF}
	\end{equation}
	where $\Psi^\init$ is some quasi-free state such that $\op^\init := \op_{\Psi^\init}$ and $\alpha^\init := \alpha_{\Psi^\init}$ satisfy the following conditions
	\begin{equation}\label{conditions:bounds_on_the_quasi-free_initial_data}
		h^d \tr_{\hh}(\op^\init)= 1 \quad \text{ and } \quad \theta_{\alpha^\init} \le C\, N^{-1/3}\, .
	\end{equation}
	Then it was proved that the quadratic in creation and annihilation operators observables of the state $\Psi_t$ are well-approximated by the BdG dynamics \eqref{eq:BdG} (with spins) in norms. 
	More precisely, the main result in \cite{marcantoni_dynamics_2023} states the following.

	\begin{theorem}[Theorem 3.3 in \cite{marcantoni_dynamics_2023}]\label{thm:mean_field_limit}
		Assume $K \in L^1(\Rd)$ and $\widehat K(\xi)\,(1+\n{\xi}^2) \in L^1(\Rd)$. Assume the initial data $\Psi^\init$ is a quasi-free state, with $\ome^\init = \op_{\Psi^\init}$ and $\al^\init =\alpha_{\Psi^\init}$, satisfying Conditions \eqref{conditions:bounds_on_the_quasi-free_initial_data}. Furthermore, assume $\ome^\init$ and $\al^\init$ satisfy the following commutator bounds: there exists $C>0$ such that 
		\begin{align}
			\sup_{\xi \in \Rd}\frac{1}{1+\n{\xi}}\Nrm{\com{e^{i\xi\cdot x}, \ome^\init}}{\L^2} \le \frac{C}{\sqrt{\hbar}}\, \quad \Nrm{\com{\grad, \ome^\init }}{\L^2}, \Nrm{\com{\grad, \al^\init }}{\L^2} \le \frac{C}{\sqrt{\hbar}}\, .
		\end{align} 
		Suppose $\Psi$ is given by Expression \eqref{eq:many-body_state} and let $(\op,\alpha)$ be a solution of the BdG dynamics \eqref{eq:BdG} with initial data $(\op^\init,\alpha^\init)$. Then there exists $\kappa_1, \kappa_2>0$, independent of $N$, such that we have the estimates for any $t\geq 0$,
		\begin{align}\label{est:op_vs_op-Psi_L2}
			\Nrm{\op_{\Psi}-\op}{\L^2(\hh)} &\le \frac{1}{\sqrt{N}}\,\exp(\kappa_1\exp(\kappa_2 \,t))\,,
			\\\label{est:al_vs_al-Psi_L2}
			\Nrm{\alpha_\Psi - \alpha}{\L^2(\hh)} &\le \frac{1}{\sqrt{N}} \exp(\kappa_1\exp(\kappa_2 \,t))\, .
		\end{align}
		where $\L^2(\hh)$ denotes the rescaled Hilbert--Schmidt norm for operators on $\hh$, also given in terms of the integral kernel by $\Nrm{\op}{\L^2(\hh)}^2 = h^d \int_{X\times X} \n{\op(\x,\y)}^2 \mu(\d\x)\,\mu(\d\y)$.
	\end{theorem}
	
	The $\L^2(\hh)$ estimates on $\alpha$ in the above theorem imply $\L^1(\hh)$ estimates for $\op_{\alpha:1}$.
	
	\begin{cor}
		For any $t\geq 0$, we have the estimate 
		\begin{equation}
			\Nrm{\op_{\alpha,\Psi:1} - \op_{\alpha:1}}{\L^1(\hh)} \leq \frac{2\,e^{2\,C_K\,h^{d-1}\,t}}{\theta_{\alpha^\init}N}\, C_t^2 + \frac{4\,e^{C_K\,h^{d-1}\,t}}{\sqrt{\theta_{\alpha^\init}}\sqrt{N}} \,C_t\,.
		\end{equation}
		where $C_t/\sqrt{N}$ is the constant appearing on the right-hand side of Inequality~\eqref{est:al_vs_al-Psi_L2} and $C_K = \Nrm{K}{L^\infty}$. Similarly, we also have 
		\begin{equation}\label{est:op_alpha_vs_op_alpha-Psi_trace-norm}
			\Nrm{\op_{\alpha,\Psi} - \op_{\alpha}}{\L^1(\hh^{\otimes 2})} := h^{2d}\tr\n{\op_{\alpha,\Psi}-\op_{\alpha}} \le \frac{e^{2\,C_K\,h^{d-1}\,t} C_t}{\sqrt{\theta_{\alpha^\init}}\sqrt{N}}\, .
		\end{equation}
		Hence, if $\theta_{\alpha^\init} \geq N^{-c}$ with $c\in[1/3,1]$,
		\begin{equation}\label{est:op_vs_op-Psi_first_marginal_L1}
			\Nrm{\op_{\alpha,\Psi:1} - \op_{\alpha:1}}{\L^1(\hh)} \leq \frac{4\,e^{2C_K\,h^{d-1}\,t}\,C_t}{N^{(1-c)/2}}\,.
		\end{equation}
	\end{cor}
	
	\begin{proof}
		Since $Nh^d=1$, it holds $\theta_\alpha\,\op_{\alpha:1} = \n{\alpha^*}^2$. Hence
		\begin{align*}
			\Nrm{\op_{\alpha,\Psi:1} - \op_{\alpha:1}}{\L^1(\hh)} &\leq \Nrm{\frac{\theta_{\Psi}\op_{\alpha,\Psi:1} - \theta_{\alpha}\op_{\alpha:1}}{\theta_\alpha}}{\L^1(\hh)} + \frac{\n{\theta_\alpha-\theta_\Psi}}{\theta_\alpha} \Nrm{\op_{\alpha,\Psi:1}}{\L^1(\hh)}
			\\
			&\leq \frac{1}{\theta_\alpha}\( \Nrm{\n{\alpha_\Psi^*}^2 - \n{\alpha^*}^2}{\L^1(\hh)} + \n{\theta_\alpha-\theta_\Psi}\)
			\\
			&\leq \frac{1}{\theta_\alpha}\( \Nrm{\alpha_\Psi^*-\alpha^*}{\L^2}^2 + 2\Nrm{\alpha_\Psi^*-\alpha^*}{\L^2} \Nrm{\alpha^*}{\L^2} + \n{\theta_\alpha-\theta_\Psi}\).
		\end{align*}
		Since $\Nrm{\alpha^*}{\L^2}^2 = \Nrm{\alpha}{\L^2}^2 = \theta_\alpha$ and
		\begin{align*}
			\n{\theta_\Psi - \theta_\alpha} &= h^d\Tr{\n{\alpha_\Psi}^2-\n{\alpha}^2} \leq \Nrm{\alpha_\Psi-\alpha}{\L^2}^2 + 2 \Nrm{\alpha_\Psi-\alpha}{\L^2} \sqrt{\theta_\alpha}
		\end{align*}
		it yields
		\begin{equation*}
			\Nrm{\op_{\alpha,\Psi:1} - \op_{\alpha:1}}{\L^1(\hh)} \leq \frac{2}{\theta_\alpha} \Nrm{\alpha_\Psi-\alpha}{\L^2}^2 + \frac{4}{\sqrt{\theta_\alpha}} \, \Nrm{\alpha_\Psi-\alpha}{\L^2}.
		\end{equation*}
		To finish the proof, notice first that by Equation~\eqref{eq:dt_theta}
		\begin{equation}
			\n{\dt\theta_\alpha} \leq 2\,h^{d-1} \Nrm{K}{L^\infty} \Nrm{\op}{\L^\infty} \theta_\alpha \leq 2\,C_K\,h^{d-1} \,\theta_\alpha
		\end{equation}
		hence $e^{-2\,C_K\,h^{d-1}\,t} \theta_{\alpha^\init} \leq \theta_\alpha$, and use the previous theorem. The proof of Inequality~\eqref{est:op_alpha_vs_op_alpha-Psi_trace-norm} is similar. 
	\end{proof}

	\subsection{\texorpdfstring{$\mathrm{SU}(2)$}{SU(2)} invariance}
	The presence of spin labels in the BdG equation complicates our studies of its semiclassical limit. To overcome this difficulty, we need to isolate out the spin labels from the BdG equation (c.f. \cite{bach_generalized_1994}). We start by noting the isomorphism $L^2(\Rd\times\{\uparrow, \downarrow\})\cong L^2(\Rd)\otimes \CC^2$. In particular, we have the identification $\cB(L^2(\Rd\times\{\uparrow, \downarrow\}))\cong \cB(L^2(\Rd))\otimes M_{2\times 2}(\CC)$ between the two spaces of bounded operators, i.e. a bounded operator $T$ acting on $L^2(\Rd\times\{\uparrow, \downarrow\})$ is identify with the matrix 
	\begin{equation}\label{def:T_matrix}
		T = 
		\begin{pmatrix}
			T_{\upuparrows} & T_{\updownarrows}\\
			T_{\downuparrows} & T_{\downdownarrows}
		\end{pmatrix}
	\end{equation}
	where $T_{\sigma\tau}$ are bounded operators acting on $L^2(\Rd)$. To factor out the spins, we further restrict ourselves to the class of $\Gamma$ operators satisfying the following $\mathrm{SU}(2)$ invariance condition in the spin space: for every $S \in \mathrm{SU}(2)$, we have that 
	\begin{equation*}
		\sfS^\ast \Gamma \sfS = \Gamma \quad \text{ where } \quad \sfS = 
		\begin{pmatrix}
			S & 0\\
			0 & \conj{S}
		\end{pmatrix}.
	\end{equation*}
	In terms of $\op$ and $\alpha$, the $\mathrm{SU}(2)$ invariance reads 
	\begin{equation*}
		S^\ast \op S = \op \quad \text{ and } \quad S^\ast \alpha \conj{S} = \alpha\, .
	\end{equation*}
	By means of elementary linear algebra, we have that $\op$ is a scalar multiple of the identity matrix and $\alpha$ is a scalar multiple of the second Pauli matrix 
	\begin{equation*}
		\sigma^{(2)} =
		\begin{pmatrix}
			0 & -i\\
			i & 0
		\end{pmatrix} ,
	\end{equation*}
	or, equivalently, we have that 
	\begin{equation}\label{eq:spin_factoring_identities}
		\op(x, \tau; x', \tau') = \op_{s}(x, x')\,\delta_{\tau\tau'} \quad \text{ and } \quad \alpha(x, \tau, x', \tau') = \alpha_s(x, x')\,\sigma^{(2)}_{\tau\tau'}\, .
	\end{equation}
	By the Pauli exclusion principle, we must have that $\alpha_s$ is symmetric, i.e. $\alpha_s(x, x')=\alpha_s(x', x)$. We also write $\op = \op_s\otimes I$ where $I$ is the $2\times 2$ identity matrix and $\alpha = \alpha_s\otimes\sigma^{(2)}$. Also, notice, by Expressions \eqref{eq:spin_factoring_identities}, the last identity of Conditions \eqref{conditions:Gamma} now reads
	\begin{equation}\label{eq:spinless_pure_quasifree_state_id}
		\n{\alpha_s^\ast}^2 = \op_s\(\id-\op_s\)\, .
	\end{equation}
	The physical meaning of the $\mathrm{SU}(2)$ invariance is discussed in \cite{hainzl_bardeencooperschrieffer_2016}.

	By Expressions \eqref{eq:spin_factoring_identities}, we write 
	\begin{equation*}
		\sfH_{\op} = \(\tfrac{\n{\opp}^2}{2} + 2\,K\ast \varrho_s(x) -h^d\sfX_{\op_s}\)\otimes I =: \sfH_{\op_s}\otimes I \quad \text{ and }\quad \sfX_\alpha\,\alpha^* = \(\sfX_{\alpha_s}\alpha_s^\ast\)\otimes I\, .
	\end{equation*}
	Then this yields the spinless equations
	\begin{subequations}\label{eq:spinless_BdG}
		\begin{align}\label{eq:spinless_BdG_gamma}
			i\hbar\,\dpt \op_s &= \com{\sfH_{\op_s},\op_s} + \sfX_{\alpha_s}\,\conj{\alpha}_s - \alpha_s\,\conj{}\sfX_{\conj{\alpha}_s}\, ,
			\\\label{eq:spinless_BdG_alpha}
			i\hbar\,\dpt \alpha_s &= \sfH_{\op_s}\alpha_s + \alpha_s \conj{\sfH}_{\op_s} + h^d\(\sfX_{\alpha_s}(\id-\conj{\op}_s)-\op_s \sfX_{\alpha_s}\) ,
		\end{align}
	\end{subequations}
	or, equivalent, in matrix form 
	\begin{equation}\label{eq:spinless_BdG_Gamma}
		i\hbar\,\dpt \Gamma_s = \com{\sfH_{\Gamma_s},\Gamma_s}
	\end{equation}
	where 
	\begin{equation*}
		\Gamma_s=	
		\begin{pmatrix}
			\op_s & \alpha_s\\
			\conj{\alpha}_s & \id-\conj{\op}_s
		\end{pmatrix}
		\quad \text{ and } \quad 
		\sfH_{\Gamma_s} =
		\(\begin{array}{cc}
			\sfH_{\op_s} & h^d\sfX_{\alpha_s}\\
			h^d\sfX_{\conj{\alpha}_s} & -\conj{\sfH}_{\op_s}
		\end{array}\)\, .
	\end{equation*}
	Notice $0\le\Gamma_s\le \id$ is self-adjoint and $\Gamma_s^2=\Gamma_s$.

	Notice the form of Equations~\eqref{eq:spinless_BdG} is almost identical to that of Equations~\eqref{eq:BdG} (except for the fact that $\alpha_s$ is symmetric and that there is a 2 in front of $K\ast\varrho_s$). In particular, we could reuse the argument in Section~\ref{sec:strategy} to obtain a semiclassical limit for Equations~\eqref{eq:spinless_BdG} since the discussion in Section~\ref{sec:strategy} is independent of the fact whether $\alpha$ is a symmetric or anti-symmetric function. In short, Theorem~\ref{thm:main} remains true for $\op_s$ and $\alpha_s$. Moreover, Inequalities~\eqref{est:op_vs_op-Psi_L2} and~\eqref{est:op_vs_op-Psi_first_marginal_L1} now read
	\begin{align*}
		&\Nrm{\op_{\Psi}-\op_s\otimes I}{\L^2(\hh)} \le \frac{C}{\sqrt{N}}\,\exp(\kappa_1\exp(\kappa_2 \,t))\,,
		\\
		&\Nrm{\op_{\alpha,\Psi} - \op_{\alpha, s}\otimes I}{\L^1(\hh^{\otimes 2})} \leq \frac{C\exp(\kappa_1\exp(\kappa_2 \,t))}{N^{(1-c)/2}}\,,
	\end{align*}
	if the initial state is $\mathrm{SU}(2)$ invariant. 

	\subsection{The joint mean-field and semiclassical limit}
		It follows from Theorem~\ref{thm:mean_field_limit} and our main result Theorem~\ref{thm:main} that one can obtain a joint limit from the many body model described above to the Vlasov equation. Indeed, we start by defining the (matrix-valued) Wigner transform for an operator $T$ of the form~\eqref{def:T_matrix} by 
		\begin{equation*}
			\mathbf{f}_T(\chi, \xi)= \intd e^{-i\xi\cdot y/\hbar}\,	
			T(\chi+\tfrac{y}{2}, \chi-\tfrac{y}{2})\d y\, ,
		\end{equation*}
		i.e. take the Wigner transform of each entry of $T$ (cf. \cite{gerard_homogenization_1997}). Then it follows from \cite[Corollary~1.1]{lafleche_quantum_2023} that
		\begin{align*}
			\Nrm{f\otimes I-\mathbf{f}_{\op_s\otimes I}}{H^{-1}(\Rdd)\otimes\CC^{2\times2}} &= 2\Nrm{f-f_{\op_s}}{H^{-1}(\Rdd)}
			\\
			&\leq 2\Wh(f,\op_s) + 2\(1+\sqrt{d}\)\sqrt{\hbar}
		\end{align*}
		where $f_{\op_s}$ is the Wigner transform of $\op_s$ and $f$ is the solution of the Vlasov equation. If $\mathbf{F}_{\op_{\alpha, s}\otimes I}$ denotes the Wigner transform of $\op_{\alpha, s}\otimes I$, i.e. $\mathbf{F}_{\op_{\alpha, s}\otimes I}$ is a $2\times 2$ matrix with entries being functions of $4d$ variables, then it also follows that 
		\begin{align*}
			\Nrm{F\otimes I-\mathbf{F}_{\op_{\alpha, s}\otimes I}}{H^{-1}(\Rdd\times\Rdd)\otimes\CC^{2\times2}} &= 2\Nrm{F-F_{\op_{\alpha, s}}}{H^{-1}(\Rdd\times \Rdd)}
			\\
			&\leq 2\Wh(F,\op_{\alpha, s}) + 2\(1+\sqrt{2d}\)\sqrt{\hbar}
		\end{align*}

		On the other hand, Theorem~\ref{thm:mean_field_limit} implies an estimate in $H^{-1}$ for the Wigner transforms since the Wigner transform is an isometry from $\L^2$ to $L^2$, and then by the continuous embedding $L^2 \subset H^{-1}$, we have that 
		\begin{equation*}
			\Nrm{\mathbf{f}_{\op_{\Psi}}-\mathbf{f}_{\op_s\otimes I}}{H^{-1}(\Rdd)\otimes\CC^{2\times2}} \le C\Nrm{\mathbf{f}_{\op_{\Psi}}-\mathbf{f}_{\op_s\otimes I}}{L^2(X)} = C\Nrm{\op_{\Psi}-\op_s\otimes I}{\L^2(\hh)}\, .
		\end{equation*}
		Similarly, by the isometry property of the Wigner transform and the quantum Sobolev inequality (see \cite[Theorem 1]{lafleche_quantum_2024}), we have that 
		\begin{equation*}
			\Nrm{\mathbf{F}_{\op_{\alpha, \Psi}}-\mathbf{F}_{\op_{\alpha, s}\otimes I}}{H^{-6}(\Rdd\times\Rdd)\otimes\CC^{2\times2}} 
			\le C\Nrm{\op_{\alpha, \Psi}-\op_{\alpha, s}\otimes I}{\L^1(\hh\otimes \hh)}\, .
		\end{equation*}
		Let us summarize the result in the following theorem.
		
		\begin{theorem}\label{thm:joint_limit}
			Let $K$ satisfies the conditions of Theorem~\ref{thm:main} and Theorem~\ref{thm:mean_field_limit}. Assume the initial data have the forms $\ome^\init=\ome^\init_s\otimes I$ and $\al^\init=\al_s^\init\otimes \sigma^{(2)}$ satisfying Conditions~\eqref{conditions:bounds_on_the_quasi-free_initial_data}. Furthermore, assume $\ome^\init_s$ and $\al^\init_s$ satisfy the following commutator bounds: there exists $C>0$ such that 
			\begin{equation*}
				\sup_{\xi \in \Rd}\frac{1}{1+\n{\xi}}\Nrm{\com{e^{i\xi\cdot x}, \ome^\init}}{\L^2} \le \frac{C}{\sqrt{\hbar}}\,, \qquad \Nrm{\com{\grad, \ome^\init }}{\L^2}, \Nrm{\com{\grad, \al^\init }}{\L^2} \le \frac{C}{\sqrt{\hbar}}\,,
			\end{equation*}   
			and $\theta_{\alpha^\init} \geq N^{-c}$ with $c\in[1/3,1]$. Let $(f, F)$ be the solutions of the system~\eqref{eq:easy_equations} with initial conditions $(f^{\init}, F^{\init})$ satisfying the conditions of Theorem~\ref{thm:main}. Then there exists constant $C, \kappa_1, \kappa_2>0$ and a polynomial function $\Lambda(t)$, independent of $N$, and such that we have the following estimates
			\begin{multline*}
				\Nrm{\mathbf{f}_{\op_{\Psi}}-f\otimes I}{H^{-1}(\Rdd)\otimes\CC^{2\times2}} \le C\,\frac{\exp(\kappa_1\exp(\kappa_2 t))}{N^{1/2}}\\
				 + \(\Wh(f^\init, \op^\init)+\Wh(F^\init, \op_\alpha^\init) + \sqrt{\hbar}\) e^{\Lambda(t)}\, ,
			\end{multline*}
			\begin{multline*}
				\Nrm{\mathbf{F}_{\op_{\alpha, \Psi}}-F\otimes I}{H^{-6}(\Rdd\times\Rdd)\otimes\CC^{2\times2}} \le C\,\frac{\exp(\kappa_1\exp(\kappa_2 t))}{N^{(1-c)/2}}\\
				+ \(\Wh(f^\init, \op^\init)+\Wh(F^\init, \op_\alpha^\init) + \sqrt{\hbar}\) e^{\Lambda(t)}\, .
			\end{multline*}
		\end{theorem}

~

\paragraph{\bf Acknowledgments.} The third author expresses gratitude to Christian Hainzl for the valuable discussions on BCS theory and for suggesting  useful references.

~

\paragraph{\bf Funding.} This work is partially supported by the Swiss National Science Foundation through the NCCR SwissMAP and the SNSF Eccellenza project PCEFP2\_181153, and by the Swiss State Secretariat for Research and Innovation through the project P.530.1016 (AEQUA).

\bibliographystyle{abbrv} 
\bibliography{Vlasov}

\end{document}